\documentclass[lettersize,journal]{IEEEtran}
\usepackage{amsmath,amssymb}
\usepackage{graphicx,graphics,color,psfrag}
\usepackage{cite,balance}
\usepackage{algorithm}
\usepackage{accents}
\usepackage{bm}
\usepackage{url}
\usepackage{algorithmic}
\usepackage[english]{babel}
\usepackage{multirow}
\usepackage{enumerate}
\usepackage{cases}
\usepackage{stfloats}
\usepackage{dsfont}
\usepackage{color,soul}
\usepackage{amsfonts}
\usepackage{tcolorbox}
\usepackage{amsmath}
\usepackage{float}
\usepackage{mathtools}
\usepackage{verbatim} %comment paragraph
\usepackage{bm,color,soul}
\usepackage{epsfig}
\usepackage{subfigure}
\usepackage{psfrag}
\usepackage{latexsym}
\usepackage[footnotesize]{caption}
\usepackage{color}
\usepackage{booktabs}
\usepackage[caption=false,font=normalsize,labelfont=sf,textfont=sf]{subfig}

\IEEEoverridecommandlockouts

\usepackage{algorithm}
\usepackage{algorithmic}
\usepackage{cuted}
\usepackage{lipsum}
\usepackage{stfloats}
\usepackage{cite,graphicx,amsmath,amssymb}
\usepackage{fancyhdr}
\usepackage{hhline}
\usepackage{array,color}
\usepackage{amsmath}
\usepackage{stfloats}
\usepackage[flushleft]{threeparttable}
\usepackage{makecell} % Xhline Xcline
\newtheorem{remark}{Remark}

\newtheorem{prop}{Proposition}

\newtheorem{proof}{proof}

\begin{document}

\title{Optimal Coordinated Transmit Beamforming for Networked Integrated Sensing and Communications}

\author{Gaoyuan~Cheng, Yuan~Fang, Jie~Xu, and Derrick Wing Kwan Ng
\thanks{Part of this paper has been presented at the IEEE International Conference on Communications (ICC) 2023 \cite{Cheng2023Coordinated}.
}
\thanks{G. Cheng, Y. Fang, and J. Xu are with the School of Science and Engineering (SSE), Future Network of Intelligence Institute (FNii), The Chinese University of Hong Kong (Shenzhen), Shenzhen 518172, China (e-mail: gaoyuancheng@link.cuhk.edu.cn, fangyuan@cuhk.edu.cn, xujie@cuhk.edu.cn). J. Xu is the corresponding author.}
\thanks{D. W. K. Ng is with the School of Electrical Engineering and Telecommunications, University of New South Wales, Sydney, NSW 2052,  Australia (e-mail: w.k.ng@unsw.edu.au).}
}

\maketitle


\begin{abstract}
This paper studies a multi-antenna networked integrated sensing and communications (ISAC) system, in which a set of multi-antenna base stations (BSs) employ the coordinated transmit beamforming to serve multiple single-antenna communication users (CUs) and perform joint target detection by exploiting the reflected signals simultaneously. To facilitate target sensing, the BSs transmit dedicated sensing signals combined with their information signals. Accordingly, we consider two types of CU receivers with and without the capability of canceling the interference from the dedicated sensing signals, respectively. In addition, we investigate two scenarios with and without time synchronization among the BSs. For the scenario with synchronization, the BSs can exploit the target-reflected signals over both the direct links (BS-to-target-to-originated BS links) and the cross-links (BS-to-target-to-other BSs links) for joint detection, while in the unsynchronized scenario, the BSs can only utilize the target-reflected signals over the direct links. For each scenario under different types of CU receivers, we optimize the coordinated transmit beamforming at the BSs to maximize the minimum detection probability over a particular targeted area, while guaranteeing the required minimum signal-to-interference-plus-noise ratio (SINR) constraints at the CUs. These SINR-constrained detection probability maximization problems are recast as non-convex quadratically constrained quadratic programs (QCQPs), which are then optimally solved via the semi-definite relaxation (SDR) technique. The numerical results show that for each considered scenario, the proposed ISAC design achieves enhanced target detection probability compared with various benchmark schemes. In particular, enabling time synchronization and sensing signal cancellation at the BSs is always beneficial for further improving the joint detection and communication performance.
\end{abstract}


\section{Introduction}

%\newpage

\IEEEPARstart{I}{ntegrated} sensing and communications (ISAC) has been recognized as one of the usage scenarios of sixth-generation (6G) wireless networks \cite{Fan2020JointRadar, Fan2023Sensing, Wei20226g, Mu2023noma,  liu2022integrated} to support emerging applications such as auto-driving, smart city, industrial automation, and unmanned aerial vehicles (UAVs) \cite{Qingqing2023UAV, mozaffari2021toward, Lyu2023Joint}. Specifically, ISAC allows the sharing of cellular base station (BS) infrastructures, signal processing modules, as well as scarce spectrum and power resources for the dual roles of wireless communications and radar sensing. This not only enhances the utilization efficiency of limited resources, but also enables seamless coordination and mutual assistance between communication and sensing for improving their performances with reduced costs. In particular, by enabling the joint optimization of the sensing and communication transmit waveforms, beamforming, and resource allocation, ISAC can efficiently harness the co-channel interference and provide new design degrees of freedom for enhancing system performance.

Conventionally, mono-static and bi-static ISAC systems have been widely investigated in the literature (see, e.g., \cite{Yuanwei2022NOMA, Zhang2022Holographic, liu2017robust, liu2018mimoradar, Khawar2015Target,  kumari2019adaptive, hua2023optimal, Fan2018MU-MIMO, YongZeng2022Waveform, Liu2020JointTransmit, Fan2018Toward, Clerckx2021Rate-Splitting} and the references therein), in which one BS serves as an ISAC transceiver or two BSs operate as the ISAC transmitter and the sensing receiver, respectively. For instance, the authors in \cite{hua2023optimal, Liu2020JointTransmit} considered the transmit beamforming in a downlink ISAC system, where a BS sends combined information-bearing and dedicated sensing signals to perform downlink multiuser communication and radar target sensing simultaneously. In particular, two joint beamforming designs were investigated in \cite{hua2023optimal}, one aimed to match the transmit beampattern with a desired one and the other aimed to maximize the transmit beampattern gains towards desired target directions, while ensuring the communication quality of service (QoS) requirements. Also, the mean square error of the beampattern as well as the cross correlation pattern were optimized in \cite{Liu2020JointTransmit} to enhance the performance of multiple-input multiple-output (MIMO) radar, subject to the communication QoS constraints. Moreover, a multi-antenna ISAC system adopting rate splitting multiple access (RSMA) was studied in \cite{Clerckx2021Rate-Splitting}, in which the joint transmission of communication streams and radar sequences was jointly optimized for optimizing the ISAC performance. However, the mono-static and bi-static ISAC systems can only offer limited service coverage and the resulting sensing and communication performances may degrade when there are rich obstacles in the environment and/or when the communication users (CUs) and sensing targets are located far apart from the BS.

Recently, motivated by multi-BS cooperation for communications (e.g., coordinated multi-point transmission/reception \cite{Gesbert2010Multi,Yu2010Coordinated}, cloud-radio access networks (C-RAN) \cite{Wu2015Cloud}, cell-free MIMO \cite{Ngo2017Cell-Free}, etc.) and distributed MIMO radar sensing \cite{Fishler2004mimoradar, Daum2009mimo, Xu2008Target, haimovich2007mimo, Shi2022Device}, the notions of networked ISAC \cite{huang2022coordinated, behdad2022power, wang2020constrained} or perceptive mobile networks \cite{Rahman2019Framework, zhang2021perceptive, JAZhang2021Uplink, Eldar2023pmn} have drawn significant research momentum to address the aforementioned issues. On the one hand, as compared with conventional cellular architectures, C-RAN and cell-free MIMO allow  centralized signal processing at the cloud to enable cooperative transmission and reception among distributed BSs, thus effectively mitigating or even exploiting the interference originating from different CUs to enhance their communication performance \cite{Gesbert2010Multi, Wu2015Cloud, Ngo2017Cell-Free}.  On the other hand, a distributed MIMO radar is able to exploit the inherent spatial diversity of target radar cross section (RCS) by designing orthogonal waveforms, thus enhancing both the sensing accuracy in estimating target parameters and the detection probability \cite{haimovich2007mimo,fishler2006spatial}. Besides, the coherent processing in MIMO radar can be further adopted to acquire high-resolution target detection exploiting both time and phase synchronization among different radar transceivers \cite{Lehmann2006High}. As such, by unifying the BSs' cooperative communications and distributed MIMO radar in integrated systems, networked ISAC is envisioned to provide seamless sensing and communication coverage, efficient interference management, enhanced communication data rate, high-resolution and high-accuracy detection and estimation \cite{Godrich2010Target}, as well as reduced energy and hardware costs.

In the literature, there have been a handful of prior works, e.g., \cite{huang2022coordinated, behdad2022power, wang2020constrained}, studying networked ISAC. For instance, the authors in \cite{huang2022coordinated} studied a single-antenna networked ISAC system, in which different BSs jointly optimized their transmit power control to minimize the total transmit power consumption, while fulfilling the required individual signal-to-interference-plus-noise ratio (SINR) constraints set by the associated CUs and the estimation accuracy or Cram{\' e}r-Rao bound (CRB) constraint for localizing a target. In addition, the work \cite{behdad2022power} considered cell-free massive MIMO for networked ISAC with regularized zero-forcing (ZF) transmit beamforming, in which the BSs jointly optimized their transmit power control over different CUs to maximize the sensing signal-to-noise ratio (SNR) while ensuring the minimum required communication SINR at the CUs. Furthermore, the work \cite{wang2020constrained} investigated the network utility maximization problem for a multi-UAV networked ISAC system, in which multiple UAVs serve a group of CUs and cooperatively sense the target simultaneously. Despite the research progress on networked ISAC, however, these prior works only considered the design of transmit power control at single-antenna BSs \cite{huang2022coordinated, wang2020constrained} or employed the simplified regularized ZF transmit beamforming at multi-antenna BSs \cite{behdad2022power}. Furthermore, these works only reused the information signals for performing target sensing \cite{huang2022coordinated, wang2020constrained} or employing only one additional common dedicated sensing beam for facilitating the target detection at a known location \cite{behdad2022power}. To the best of our knowledge, the joint optimization of transmit information beamforming and general-rank transmit sensing beamforming in multi-antenna networked ISAC systems has not been well investigated in the literature yet. This task, however, is particularly challenging due to the following reasons. First, while the information and sensing signals at multiple BSs may cause interference at different CUs, these signals can also be jointly exploited for performing target sensing. As a result, this introduces a new tradeoff between mitigating the co-channel interference to enhance the communication performance and increasing the sensing signal strength to improve the sensing performance. Next, practical dedicated sensing signals can be generated offline and are known to each CU prior to transmission \cite{hua2023optimal}. Therefore, this presents a new opportunity to exploit interference cancellation to enhance the SINR at the CUs. However, the impact of such interference cancellation on the networked ISAC has not been investigated thus far. Furthermore, distributed MIMO sensing may exploit the cross-link echo signals from one sensing transmitter to a target captured by another sensing receiver for facilitating sensing \cite{haimovich2007mimo}, when perfect synchronization can be achieved between them. Indeed, investigating the impact of such synchronization in the performance of networked ISAC is an intriguing problem that remains unexplored. To address the above issues, the joint design of target detection and multiuser communication in networked ISAC are of utmost importance.

This paper studies a multi-antenna networked ISAC system, in which a set of multi-antenna BSs employ coordinated transmit beamforming to serve their associated CUs and at the same time reuse the reflected wireless signals to perform joint target detection. Our main results are summarized as follows.

\begin{itemize}
\item To fully utilize the degrees of freedom for sensing, the BSs sends dedicated sensing signals in addition to communication signals. Accordingly, to exploit the benefit of these newly introduced dedicated sensing signals, we consider two types of CU receivers: those without the capability of canceling the interference from dedicated sensing signals (Type-I receivers) and those with the capability to perform that (Type-II receivers), respectively. %Furthermore, we rigorously proofed that there is no need for adding dedicated radar signal when the CUs are Type-I receivers.

\item We consider two target detection scenarios depending on the availability of time synchronization among the BSs. In Scenario I, these BSs are all synchronized in time such that they can exploit the target-reflected signals over both the direct links (BS-to-target-to-originated BS links) and the cross links (BS-to-target-to-other BSs links) for joint detection. In Scenario II, these BSs are not synchronized and thus they can only utilize the target-reflected signals over their direct links for joint detection. For each of the two scenarios, we analyze the likelihood ratio test for detection and accordingly derive the detection probability subject to a required false alarm probability at any given target location, showing that the detection probability is monotonically increasing with respect to the total received reflection-signal power (over the utilized links for each scenario) at the BSs.

\item Based on the derivation in each scenario and by considering each type of CU receivers, we propose the coordinated transmit beamforming design at the BSs to maximize the minimum detection probability (or equivalently the minimum total received reflection-signal power) over a particular targeted area, while satisfying the minimum SINR constraints at the CUs, subject to the maximum transmit power constraints at the BSs. These problems are recast as non-convex quadratically constrained quadratic programs (QCQPs), which are then optimally solved via the semi-definite relaxation (SDR) technique. In particular, we rigorously prove that the adopted SDRs are tight for these QCQPs and the optimal rank-one solutions for information beamforming can be properly constructed based on the optimal solution of the SDRs.

\item Finally, we provide numerical results to validate the performance of our proposed designs as compared to two benchmark schemes that perform ZF information beamforming and conduct the target detection only via dedicated sensing signals, respectively. It is shown that for each scenario, the proposed ISAC design achieves a higher detection probability than the benchmark schemes. It is also shown that ensuring time synchronization among BSs in Scenario I is consistently enhances the detection performance. Moreover, under both Scenario I and Scenario II, we show that Type-II CUs equipped with sensing interference cancellation capability outperform their Type-I counterparts and other benchmark designs in terms of detection performance, due to the higher flexibility in interference management of the Type-II CUs, which is always beneficial.
\end{itemize}

The remainder of this paper is organized as follows. Section II presents the networked ISAC system model. Section III derives the detection probability under a specific false alarm probability at any target location. Section IV presents the optimal coordinated transmit beamforming optimization problems for the two considered scenarios, respectively. Section V provides numerical results to validate the performance of our proposed schemes. Section VI concludes this paper.

{\textit{Notations:}} Vectors and matrices are denoted by boldface lowercase and uppercase letters, respectively. $\bf{I}$ denotes an identity matrix with appropriate dimension. $\mathbb E (\cdot)$ denotes the statistical expectation. ${\mathop{\rm var}} \left(  \cdot  \right)$ denotes the statistical variance. For a scalar $a$, $\left| a \right|$ denotes its absolute value. For a vector $\bf v$, $\left\| {\bf{v}} \right\|$ denotes its Euclidean norm. For a matrix ${\bf M}$ of arbitrary dimension, ${\bf M}^T$ and ${\bf M}^H$ denote its transpose and conjugate transpose, respectively. ${\mathbb C}^{x \times y}$ denotes the space of $x \times y$ complex matrices. ${\mathop{\rm Re}\nolimits} \left(  \cdot  \right)$ denotes the real part of a complex number, vector, or matrix. $\cal{N}({\bf x}, {\bf Y})$ and $\cal{CN}({\bf x}, {\bf Y})$ denote the real-valued Gaussian and the circularly symmetric complex Gaussian (CSCG) distributions with mean vector ${\bf x}$ and covariance matrix ${\bf{Y}}$, respectively, and ``$ \sim $" means ``distributed as". $Q(\cdot)$ denotes the Q-function. ${\rm{rank}}\left( {\cdot} \right)$ denotes the rank of a matrix.

\begin{figure}[ht]
\centering
    \includegraphics[width=8cm]{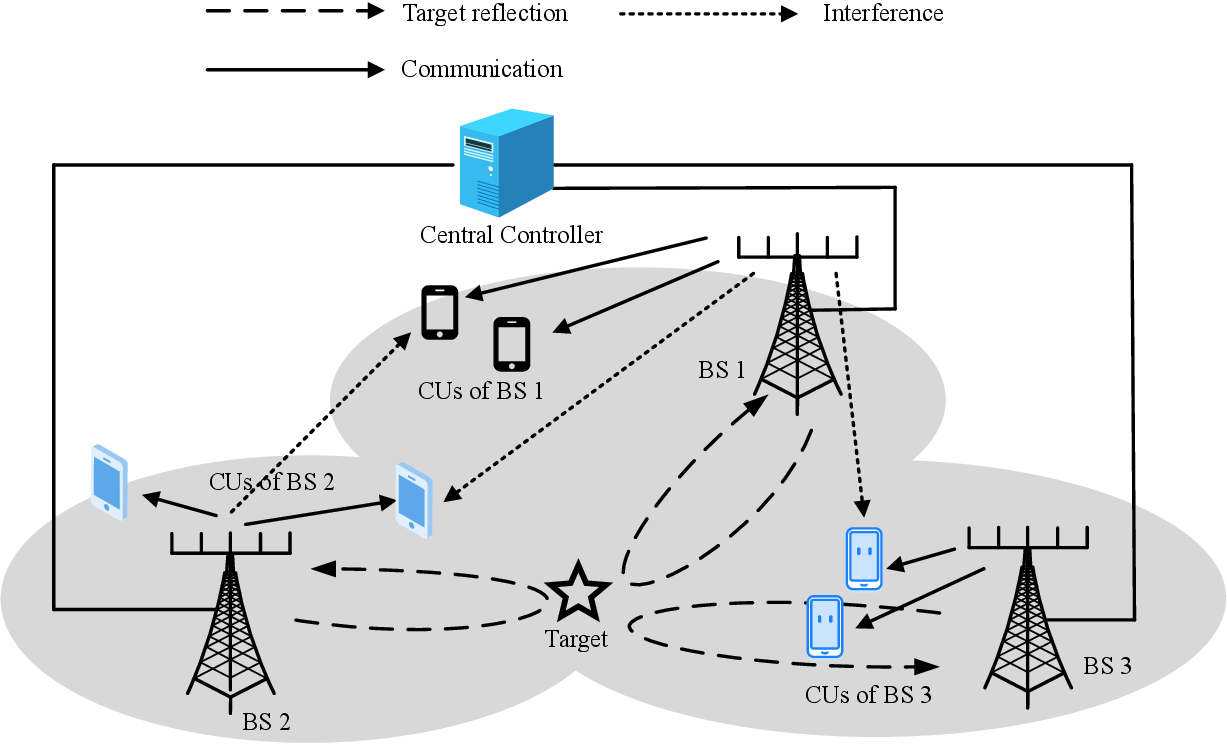}
\caption{An example of the considered multi-antenna networked ISAC system model with three coordinated BSs.}
\label{fig:1}
\end{figure}

\section{System Model}

We consider a multi-antenna networked ISAC system consisting of $L$ BSs each with $N_t>1$ transmit and $N_r>1$ receive antennas, where each BS serves the same number of $K$ single-antenna CUs. Note that the total number of CUs is $LK$,  let $ {\cal L} \triangleq \left\{ {1, \ldots ,L} \right\}$ and $ {\cal K}_l \triangleq \left\{ {1, \ldots ,K} \right\}$ denote the set of BSs and the set of CUs in each cell associated with BS $l$, respectively. In this system, the BSs send individual messages and dedicated sensing signals to their associated CUs. At the same time, the BSs receive and properly process the reflected signals and then convey them to a central controller (CC) for joint target detection, cf. Fig. 1. As such, the multi-antenna networked ISAC system unifies the multi-antenna coordinated beamforming system for communication \cite{Yu2010Coordinated} and the distributed MIMO radar for target detection \cite{haimovich2007mimo}, as will be detailed next. Specifically, we focus on the ISAC transmission over a communication block with duration $T$ that consists of $N$ symbols, where $T = N T_{\mathrm s}$ with $T_{\mathrm s}$ denoting the duration of each symbol. Here, $T$ or $N$ is assumed to be sufficiently large for the ease of analysis \cite{huang2022coordinated}. Let $\mathcal T \triangleq (0,T]$ denote the ISAC period of interest and $\mathcal N \triangleq \{1, \ldots, N\}$ the set of symbols.

%= \sum\limits_{n = 1}^{{r_t}} {{{\bf{t}}_{l,n}}\bar s_{l,n}^r}(t), 1 \le r_t \le N_t
First, we consider the communication from the BSs to the CUs, in which the coordinated transmit beamforming is employed at these BSs. Let $\bar s_{l,i}\left(t\right) \in {\mathbb C}$ denote the communication signal sent by BS $l \in {\cal L}$ for CU $i \in {\cal K}_l$ at time $t \in {\cal T}$, ${\bf{w}}_{l,i} \in {{\mathbb C}^{N_t\times 1}}$ denote the corresponding transmit beamforming vector by BS $l \in \mathcal L$, and ${\bf{\bar s}}_l^r\left( t \right) {{\mathbb C}^{N_t\times 1}}$ denote the dedicated sensing signal sent by BS $l$ at time $t$ with zero mean and covariance matrix ${\bf{R}}_l^r = {\mathbb E}[{\bf{\bar s}}_l^{r}\left( t \right){\bf{\bar s}}_l^{rH}\left( t \right)] \succeq {\bf 0}$. With loss of generality, we assume $r_t = {\text{rank}}({\bf R}_l^r)$, As such, we express ${\bf{\bar s}}_l^r \left( t \right)$ as
\begin{align}
{\bf{\bar s}}_l^r\left( t \right)= \sum\limits_{k = 1}^{{r_t}} {{{\bf{w}}_{l,k}^r}\bar s_{l,k}^r}(t),
\end{align}
where ${\bar s_{l,k}^r}(t)$ denotes the $k$th waveform of BS $l$ that is modeled by an independently generated pseudorandom signal with zero mean and unit variance, and ${\bf w}_{l,k}^r$ denotes the corresponding transmit beamforming vector that can be determined based on ${\bf{R}}_l^r$ via eigenvalue decomposition (EVD). Denote $s_{l,i}[n]$, ${s}_{l,k}^r[n]$, and ${\bf{s}}_l^r[n]$ as the sampled signals of $\bar s_{l,i}\left(t\right)$, ${{\bar s}}_{l,k}^r\left( t \right)$, and $\bar{ \bf{s}}_l^r\left( t \right)$, respectively, at each symbol $n\in \mathcal N$. Here, $\{{s}_{l,i}[n]\}$ are assumed to be independent and identically distributed (i.i.d.) random variables with zero mean and unit variance.

Let ${{\bf{h}}_{l,m,i}}\in \mathbb{C}^ {N_t\times 1}$ denote the channel vector from BS $l\in \mathcal L$ to CU $i\in {\mathcal K}_m$ that is located in cell $m \in \cal L$. Then, the received signal by CU $k \in {\mathcal K}_m$ at cell $m$ in symbol $n \in \mathcal N$ is
\begin{align}
{y_{m,k}}[ n ] &= {\bf{h}}_{m,m,k}^H{{\bf{w}}_{m,k}}{s_{m,k}}[ n ] + \! \underbrace {\sum\limits_{i \in {\cal K}_m,i \ne k} \! \! {{\bf{h}}_{m,m,k}^H{{\bf{w}}_{m,i}}{s_{m,i}}[n]} }_{{\rm{intra-cell~ interference}}}\nonumber \\ &+\underbrace {\sum\limits_{l \in {\cal L},l \ne m} {{\bf{h}}_{l,m,k}^H\sum\limits_{i \in {\cal K}_l} {{{\bf{w}}_{l,i}}{s_{l,i}}[n]} } }_{{\rm{inter-cell~interference}}}+ \underbrace {\sum\limits_{l \in {\cal L}} {{\bf{h}}_{l,m,k}^H{\bf{s}}_l^r[n]} }_{{\rm{sensing~interference}}} \nonumber \\ &+ {z_{m,k}}[ n ], \label{comrecsig}
\end{align}
where ${z_{m,k}}[ n ]  \sim  {\cal{CN}} \left( {0,\sigma ^2_c} \right)$ denotes the noise at the receiver of CU $k$ at cell $m$, with $\sigma_c^2$ denoting the corresponding noise power. It is observed in \eqref{comrecsig} that each CU suffers from both the intra-cell and inter-cell interference, as well as the interference from the dedicated sensing signals. In practice, the noise term ${z_{m,k}}[ n ]$ may also include the background and clutter interference \cite{xianxin2023ICC}. Notice that $\{{\bf {\bar s}}_{l}^r(t)\}$ are predetermined pseudorandom signals that can be {\textit{a-priori}} known by all the BSs and CUs. Therefore, we consider two different types of CUs: those without the capability of canceling the interference caused by the dedicated sensing signal \cite{hua2023optimal}, referred to as Type-I receivers, and those with the capability, referred to as Type-II receivers, respectively.

{\bf {Type-I receivers}}: Each Type-I receiver $k$ in cell $m$ is not equipped with the capability to cancel the interference generated by the dedicated sensing signal ${\bf{s}}_l^r[n]$. The SINR of CU $k$ in cell $m$ is given by \eqref{sinr:I}.
\begin{figure*}
\begin{align}
\gamma _{m,k}^{\rm I} \left( {\left\{ {{{\bf{w}}_{l,i}}} \right\},\left\{ {{\bf{R}}_l^r} \right\}} \right) = \frac{{{{\left| {{\bf{h}}_{m,m,k}^H{{\bf{w}}_{m,k}}} \right|}^2}}}{{\sum\limits_{i \in {\cal K}_m, i \ne k} {{{\left| {\bf{h}}_{m,m,k}^H{{\bf{w}}_{m,i}} \right|}^2}}  + \sum\limits_{l \in {\cal L},l \ne m} {\sum\limits_{i \in {\cal K}_l} \left|{{{\bf{h}}_{l,m,k}}{{\bf{w}}_{l,i}}}\right|^2 }  + \sum\limits_{l \in {\cal L}} {{\bf{h}}_{l,m,k}^H{\bf{R}}_l^r{{\bf{h}}_{l,m,k}}}  + \sigma _c^2}}. \label{sinr:I}
\end{align}
%\hrulefill
\end{figure*}

{\bf {Type-II receivers}}: Each Type-II receiver $k$ in cell $m$ is dedicatedly designed for the ISAC system with the capability to cancel the interference generated by dedicated sensing signal ${\bf{s}}_l^r[n]$ before decoding its desired communication signal ${{s}}_{m,k}[n]$. In this case, the SINR of CU $k$ in cell $m$ is given by \eqref{sinr:II}.
\begin{figure*}
\begin{align}
\gamma _{m,k}^{\rm II} \left( {\left\{ {{{\bf{w}}_{l,i}}} \right\},\left\{ {{\bf{R}}_l^r} \right\}} \right) = \frac{{{{\left| {{\bf{h}}_{m,m,k}^H{{\bf{w}}_{m,k}}} \right|}^2}}}{{\sum\limits_{i \in {\cal K}_m,i \ne k} {{{\left| {{\bf{h}}_{m,m,k_m}^H{{\bf{w}}_{m,i_m}}} \right|}^2}}  + \sum\limits_{l \in {\cal L},l \ne m} \sum\limits_{i \in {\cal K}_l} \left|{\bf{h}}_{l,m,k}^H{\bf{w}}_{l,i}\right|^2   + \sigma _c^2}}. \label{sinr:II}
\end{align}
\hrulefill
\end{figure*}

Next, we consider the distributed MIMO radar detection by the $L$ BSs via reusing both the communication signals $\{\bar s_{l,i}(t)\}$ and the dedicated sensing signals $\{{\bf {\bar s}}_l^r(t)\}$ concurrently. Let $(x_{l}, y_{l})$ denote the location of each BS $l\in\cal{L}$. Suppose that there is one target present at location $(x_0,y_0)$, for which the target angle with respect to BS $l$ is denoted by $\theta_l$. Let ${{\bf{a}}_{t,l}}\left( {{\theta _l}} \right) \in \mathbb{C}^ {N_t\times 1} $ and ${\bf{a}}_{r,l}\left( {{\theta _l}} \right) \in \mathbb{C}^ {N_r\times 1} $ denote the transmit and receive steering vectors at BS $l\in\mathcal L$, respectively, where ${{\left\| {{{\bf{a}}_{t,l}}({\theta _l})} \right\|} \mathord{\left/
 {\vphantom {{\left\| {{{\bf{a}}_{t,l}}({\theta _l})} \right\|} {{N_t}}}} \right.
 \kern-\nulldelimiterspace} {{N_t}}} = {{\left\| {{{\bf{a}}_{r,l}}({\theta _l})} \right\|} \mathord{\left/
 {\vphantom {{\left\| {{{\bf{a}}_{r,l}}({\theta _l})} \right\|} {{N_r}}}} \right.
 \kern-\nulldelimiterspace} {{N_r}}} = 1$ is assumed without loss of generality \cite{hua2023optimal}. In the practical case with uniform linear arrays (ULAs) deployed at the BSs, we have
\begin{align}
 {{\bf{a}}_{t,l}}\left( {{\theta _l}} \right) = {\left[ {1,{e^{j2\pi \frac{{{d_a}}}{\lambda }\sin \left( {{\theta _l}} \right)}}, \ldots ,{e^{j2\pi \frac{{{d_a}}}{\lambda }\left( {N_t - 1} \right)\sin \left( {{\theta _l}} \right)}}} \right]^T},
\end{align}
and
\begin{align}
{{\bf{a}}_{r,l}}({\theta _l}){\rm{  = }}{\left[ {1,{e^{j2\pi \frac{{{d_a}}}{\lambda }\sin \left( {{\theta _l}} \right)}}, \ldots ,{e^{j2\pi \frac{{{d_a}}}{\lambda }\left( {{N_r} - 1} \right)\sin \left( {{\theta _l}} \right)}}} \right]^T},
\end{align}
where $j =\sqrt{-1}$, while parameters $d_a$ and $\lambda$ denote the antenna spacing and wavelength, respectively. Let ${{\bf{H}}_{m,l}} = {{ \hat\zeta }_{m,l}}{\bf{a}}_{r,m}\left( {{\theta _m}} \right){{\bf{a}}_{t,l}^T}\left( {{\theta _l}} \right) \in \mathbb{C}^ {N_r\times N_t} $ denote the end-to-end target response matrix from BS $l$-to-the-target-to-BS $m$, in which ${{\hat \zeta }_{m,l}} = \sqrt {{\beta _{m,l}}} {{\zeta} _{m,l}}$ is the reflection coefficient incorporating both the RCS ${{ \zeta }_{m,l}}$ and the equivalent  round-trip path loss ${\beta _{m,l}}$. Specifically, we assume that ${\beta _{m,l}} = \kappa^2 { {\frac{{d_{\rm ref}^4}}{{{d_{m}^2}{d_{l}^2}}}} }$, where $\kappa$ denotes the path loss at the reference distance $d_{\rm ref}$ and ${d_{l}} = \sqrt {({x_l} - {x_0}){}^2 + ({y_l} - {y_0}){}^2} $ denotes the distance between the target and BS $m$. As such, the received echo signal at BS $m$ is  %We also assume an isotropic target with uniform RCS, i.e., ${\zeta _{m,l}}=\zeta , \forall m,l \in \cal L$. Under this setup, the received target-reflection signal at BS $m$ is given by
\begin{align}
{\bf{r}}_m\left( t \right) &= \sum\limits_{l \in {\cal L}} {\bf{H}}_{m,l}\left( \sum\limits_{i \in {\cal K}_l} {\bf{w}}_{l,i}{\bar s}_{l,i}\left( t - \tau _{m,l} \right)  + { \bf{\bar s}}_l^r\left( t - \tau _{m,l} \right) \right) \nonumber \\  &+ {\bf{\bar z}}_m\left( t \right), \label{recsig}
\end{align}
where ${\bf{\bar z}}_m(t) \sim  {\cal{CN}} \left( {{\bf 0},\sigma ^2_d{\bf{I}}} \right)$ denotes the noise at the receiver of BS $m$ and ${\tau_{m,l}} = \frac{1}{c}({d_{m}} + {d_{l}})$ denotes the transmission delay from BS $l$-to-the-target-to-BS $m$, with $c$ denoting the speed of light. Without loss of generality, we assume that each information signal and dedicated sensing signal have a normalized power over block $\mathcal T$, i.e., $\frac{1}{T}\int_\mathcal T {{{\left| {\bar s_{l,i}}\left( t \right) \right|}^2}dt}  = 1$ and $\frac{1}{T}\int_{\cal T} {{{\left| {\bar s_{l,k}^r\left( t \right)} \right|}^2}dt}  = 1$. Furthermore, notice that $\{\bar s_{l,i}(t)\}$ and $\{\bar s_{l,k}^r(t)\}$ are with zero mean and independent over different CUs and different times. As $T$ is sufficiently large, we have $ \frac{1}{T}\int_\mathcal T {\bar s_{l,i}\left( t \right)\bar s_{m,k}^*\left( {t - \tau } \right)dt}= 0, \forall \tau, l\neq m, k\neq i$, and $ \frac{1}{T}\int_{\cal T} {\bar s_{l,k}^r\left( t \right)\bar s_{m,i}^{r*}\left( {t - \tau } \right)dt}  = 0, \forall \tau, l\neq m, k\neq i$, as well as $\frac{1}{T} \int_\mathcal T \bar s_{l,i}(t) \bar s_{l,i}^* (t-\tau) dt = 0$ and  $ \frac{1}{T}\int_{\cal T} {\bar s_{l,i}^r} (t)\bar s_{l,i}^{r*}(t - \tau )dt = 0 $ for any $l, i$ and $|\tau| \ge T_{\mathrm s}$. Based on the received signals $\{{\bf r}_m(t)\}$ in (\ref{recsig}), the $L$ BSs jointly detect the existence of the target, as will be illustrated in the next section.

\section{Detection Probability at Given Target Location }

In this section, we derive the detection probability and the false alarm probability at a given target location $(x_0, y_0)$, by particularly considering the two joint detection scenarios with and without time synchronization among the BSs, namely Scenario-I and Scenario-II, respectively.

{\bf Scenario \uppercase\expandafter{\romannumeral1}}: All the BSs are synchronized in time such that the mutual delays $\tau_{m,l}$ are known. As such, the target-reflected signals over both the direct links (i.e., ${\bf{H}}_{m,m}( \sum\limits_{i \in {\cal K}_m} {\bf{w}}_{m,i}{\bar s}_{m,i}( {t - {\tau _{m,m}}} )  + {\bf{\bar s}}_m^r( t - \tau _{m,m} ) )$ from each BS $m$-to-the-target-to-originated BS) and the cross links (i.e., ${\bf{H}}_{m,l}( \sum\limits_{i \in {\cal K}_l} {{\bf{w}}_{l,i}}{\bar s}_{l,i}( t - \tau _{m,l} )  + {\bf{\bar s}}_l^r( t - \tau _{m,l} ) )$ from other BSs $l$'s-to-the-target-to-BS $m$, $\forall l\neq m$) can be exploited for joint detection. Towards this end, each BS $m$ performs the matched filtering (MF) processing based on ${\bf{r}}_m(t)$ by exploiting $\{\bar s_{l,i}(t)\}$, $\{\bar s_{l,k}^r(t)\}$, and delay $\{\tau_{m,l}\}$. Accordingly, the processed signal based on $\bar{s}_{l,i}(t)$ is
\begin{align}
{\bf d}_{m,l,i} &=\frac{1}{T}\int_\mathcal T {{{\bf{r}}_m}\left( t \right)\bar s_{l,i}^*\left( {t - {\tau _{m,l}}} \right)dt} \nonumber \\ &= \underbrace {\frac{1}{T}\int_\mathcal T {{{\bf{H}}_{m,l}}{{\bf{w}}_{l,i}}{{\left| {\bar s_{l,i}^*\left( {t - {\tau _{m,l}}} \right)} \right|}^2}dt}  }_{\rm desired~signal} \nonumber \\ &+ \underbrace { \frac{1}{T}\int_\mathcal T {{{\bf{z}}_{m}}\left( t \right)\bar s_{l,i}^*\left( {t - {\tau _{m,l}}} \right)dt}  }_{\rm filtered ~noise} \nonumber \\
& = {\bf{H}}_{m,l}{\bf{w}}_{l,i} + {\hat {\bf z}}_{m,l,i}.
\label{mod:1}
\end{align}
Similarly, the processed signal based on $\bar s_{l,k}^r(t)$ is
\begin{align}
{\bf d}_{m,l,k}^r &=\frac{1}{T}\int_\mathcal T {\bf{r}}_m\left( t \right)\bar s_{l,k}^{r*}\left( t - {\tau _{m,l}} \right)dt \nonumber \\ &= \underbrace {\frac{1}{T}\int_\mathcal T {\bf{H}}_{m,l}{\bf{w}}_{l,k}^r{{\left| {\bar s_{l,k}^{r*}\left( {t - {\tau _{m,l}}} \right)} \right|}^2}dt}_{\rm desired~signal} \nonumber \\ &+ \underbrace { \frac{1}{T}\int_\mathcal T {\bf{z}}_m\left( t \right)\bar s_{l,k}^{r*}\left( t - \tau _{m,l} \right)dt  }_{\rm filtered ~noise}\nonumber \\
&= {\bf{H}}_{m,l}{\bf{w}}_{l,k}^r + {\hat {\bf z}}_{m,l,k}^r.
\label{mod:2}
\end{align}
In \eqref{mod:1} and \eqref{mod:2}, ${\hat {\bf z}_{m,l,i}} \sim {\cal CN}\left({{\bf 0},{\sigma ^2_d}{\bf I}} \right)$ and ${\hat {\bf z}_{m,l,k}^r} \sim {\cal CN}\left({{\bf 0},{\sigma ^2_d}{\bf I}} \right)$ denote the equivalent noise after MF processing, where $\sigma_d^2$ is the per-antenna noise power at each BS. After obtaining $\{{\bf d}_{m,l,i}\}_{l\in\mathcal{L}}$ and $\{{\bf d}_{m,l,k}^r\}_{l\in\mathcal{L}}$, each BS $m$ delivers them to the CC, which then performs the joint radar detection based on $\{{\bf d}_{m,l,i}\}$ and $\{{\bf d}_{m,l,k}^r\}$. For ease of illustration, the observed signals can be stacked as \eqref{stack:d_I}.
\begin{figure*}
\begin{align}
{{\bf{d}}_{\rm{I}}} = {\left[ {{\bf{d}}_{1,1,1}^T, \ldots ,{\bf{d}}_{1,1,K}^T,{\bf{d}}_{1,2,1}^T, \ldots ,{\bf{d}}_{L,L,K}^T,{\bf{d}}_{1,1,1}^{rT}, \ldots ,{\bf{d}}_{1,1,{r_t}}^{rT},{\bf{d}}_{1,2,1}^{rT}, \ldots ,{\bf{d}}_{L,L,{r_t}}^{rT}} \right]^T} \in {{\mathbb C}^{{N_r(K+ r_t)L^2}}}. \label{stack:d_I}
\end{align}
%\hrulefill
\end{figure*}

{\bf Scenario \uppercase\expandafter{\romannumeral2}}: In this scenario, the BSs are not synchronized in time and thus the value of transmission delay $\tau _{m,l}$ with respect to other BS $l\neq m$ is not available at BS $m$. In this scenario, the BSs can only utilize the target-reflected signals over their direct links, i.e., (${\bf{H}}_{m,m}( \sum\limits_{i \in {\cal K}_m} {\bf{w}}_{m,i}{\bar s}_{m,i}( {t - {\tau _{m,m}}} )  + {\bf{s}}_m^r( t - \tau _{m,m} ) )$, $\forall m\in\mathcal L$), for joint detection. After the MF processing similarly as in Scenario I, we have the processed signal as  \eqref{stack:d_II},
\begin{figure*}
\begin{align}
{\bf{d}}_{\rm II} = {\left[ {{\bf{d}}_{1,1,1}^T, \ldots ,{\bf{d}}_{1,1,K}^T,{\bf{d}}_{2,2,1}^T, \ldots ,{\bf{d}}_{L,L,K}^T,{\bf{d}}_{1,1,1}^{rT}, \ldots ,{\bf{d}}_{1,1,{r_t}}^{rT},{\bf{d}}_{2,2,1}^{rT}, \ldots ,{\bf{d}}_{L,L,{r_t}}^{rT}} \right]^T} \in {{\mathbb C}^{{N_r(K+ r_t)L}}}. \label{stack:d_II}
\end{align}
\hrulefill
\end{figure*}
where $\{{\bf d}_{m,l,i}\}$ and $\{{\bf d}_{m,l,k}^r\}$ are defined as \eqref{mod:1} and \eqref{mod:2}, respectively.

\subsection{Detection Probability in Scenario \uppercase\expandafter{\romannumeral1} with BSs Synchronization}\label{Sec:III-A}

To start with, we define two hypotheses for target detection, i.e., ${\cal H}_1$ when the target exists and ${\cal H}_0$ when the target does not exist. For notational simplicity, we define ${\boldsymbol \alpha} _{m,l,i}^c = {\bf{H}}_{m,l}{\bf{w}}_{l,i}$ and ${\boldsymbol \alpha} _{m,l,k}^r = {\bf{H}}_{m,l}{\bf{w}}_{l,k}^r$ as the reflected communication signal and dedicated sensing signal vectors from BS $l$-to-the target-to-BS $m$ when the target exists. Also, we present the correspondingly accumulated signal vector as \eqref{alpha:I}.
\begin{figure*}
\begin{align}
{\boldsymbol{\alpha }}_{\rm I} = {\left[ {{\boldsymbol{\alpha }}_{1,1,1}^{cT}, \ldots, {\boldsymbol{\alpha }}_{1,1,K}^{cT},{\boldsymbol{\alpha }}_{1,2,1}^{cT}, \ldots ,{\boldsymbol{\alpha }}_{L,L,K}^{cT},{\boldsymbol{\alpha }}_{1,1,1}^{rT}, \ldots ,{\boldsymbol{\alpha }}_{1,1,{r_t}}^{rT},{\boldsymbol{\alpha }}_{1,2,1}^{rT}, \ldots ,{\boldsymbol{\alpha }}_{L,L,{r_t}}^{rT}} \right]^T} \in {{\mathbb C}^{{N_r(K+r_t)L^2}}}. \label{alpha:I}
\end{align}
\end{figure*}
Furthermore, define the noise vector in Scenario I as \eqref{noise:I}.
\begin{figure*}
\begin{align}
{\hat{\bf z}}_{\rm I} = {\left[ {{\hat{\bf z}}_{1,1,1}^{T}, \ldots ,{\hat{\bf z}}_{1,1,K}^{T},{\hat{\bf z}}_{1,2,1}^{T}, \ldots ,{\hat{\bf z}}_{L,L,K}^{T},{\hat{\bf z}}_{1,1,1}^{rT}, \ldots ,{\hat{\bf z}}_{1,1,{r_t}}^{rT},{\hat{\bf z}}_{1,2,1}^{rT}, \ldots ,{\hat{\bf z}}_{L,L,{r_t}}^{rT}} \right]^T} \in {{\mathbb C}^{{N_r(K+r_t)L^2}}}. \label{noise:I}
\end{align}
\hrulefill
\end{figure*}
Then, based on \eqref{mod:1}, we have the processed signals after the MF processing as
\begin{align}\label{eqn:H_1:H_0}
\left\{ \begin{array}{l}
{{\cal H}_1}:{{\bf{d}}_{\rm{I}}} = {{\boldsymbol{\alpha }}_{\rm{I}}} + {\hat {\bf{z}}_{\rm{I}}},\\
{{\cal H}_0}:{{\bf{d}}_{\rm{I}}} = {\hat {\bf{z}}_{\rm{I}}}.
\end{array} \right.
\end{align}

%when  holds, otherwise there only exist the filtered noise, which refers to

%First,  we define two hypothesis for detection, which are given by
%\begin{align}
%\left\{ {\begin{array}{*{20}{c}}
%{{\cal H}_1}: {{\bf{d}}_{k,i}} = {{\bf{H}}_{k,i}}{{\bf{w}}_i} + {{\hat {\bf z}}_{k,i}}, \forall k,i \in \cal K\\
%{{\cal H}_0}: {{\bf{d}}_{k,i}}=  {{\hat {\bf z}}_{k,i}}, \forall k,i \in \cal K
%\end{array}} \right.
%\end{align}

Next, we adopt the likelihood ratio test for target detection. Based on \eqref{eqn:H_1:H_0}, the likelihood functions of vector ${{\bf{d}}}_{\rm I}$ under the hypothesis ${{\cal H}_1}$ and ${{\cal H}_0}$ are respectively given by
\begin{align}
{p\left( {{\bf{d}}_{\rm I}|{{\cal H}_1}} \right)} &\!= \! {c_0\exp \left( { \!- \frac{1}{{{\sigma ^2_d}}}{{\left( {{\bf{d}}_{\rm I} \!-\! {\boldsymbol{\alpha }}_{\rm I}} \right)}^H}\!\left( {{\bf{d}}_{\rm I} \!-\! {\boldsymbol{\alpha }_{\rm I}}} \right)} \right)}, \\
{p\left( {{\bf{d}_{\rm I}}|{{\cal H}_0}} \right)} & = {c_0\exp \left( { - \frac{1}{{{\sigma ^2_d}}}{{\bf{d }}_{\rm I}^H}{\bf{d }}_{\rm I}} \right)},
\end{align}
where $c_0  =  \frac{1}{{{\pi ^{{N_r(K+r_t)L^2}}}{\sigma ^{2{N_r(K+r_t)L^2}}_d}}}$.
Accordingly, the Neyman-Pearson (NP) detector is given by the likelihood ratio test\cite{sssssss}:
\begin{align}
\ln  \frac{{p\left( {{\bf{d}}_{\rm I}|{{\cal H}_1}} \right)}}{{p\left( {{\bf{d}}_{\rm I}|{{\cal H}_0}} \right)}} %\nonumber \\
%&= \frac{{\frac{1}{{{\pi ^{{N_rK^2}}}{\sigma ^{2{N_rK^2}}_d}}}\exp \left( { - \frac{1}{{{\sigma ^2_d}}}{{\left( {{\bf{d}} - {\boldsymbol{\alpha }}} \right)}^H}\left( {{\bf{d}} - {\boldsymbol{\alpha }}} \right)} \right)}}{{\frac{1}{{{\pi ^{{N_rK^2}}}{\sigma ^{2{N_rK^2}}_d}}}\exp \left( { - \frac{1}{{{\sigma ^2_d}}}{{\boldsymbol{\alpha }}^H}{\boldsymbol{\alpha }}} \right)}}
= \frac{1}{{{\sigma_d ^2}}}\left( {2{\rm{Re}}\left( {{{\boldsymbol{\alpha }}_{\rm I}^H}{\bf{d}}}_{\rm I} \right) - {{\boldsymbol{\alpha }}_{\rm I}^H}{\boldsymbol{\alpha }}_{\rm I}} \right) \mathop \gtrless\limits_{{\cal H}_0}^{{\cal H}_1}  \delta,\label{mod:2:v2}
\end{align}
where $\delta$ denotes the threshold determined by the tolerable level of false alarm. Notice that since ${{\boldsymbol{\alpha }}_{\rm I}^H}{\boldsymbol{\alpha }}_{\rm I}$ is given, the detector in \eqref{mod:2:v2} can be equivalently simplified as
\begin{align}\label{detector}
T\left( {\bf{d}}_{\rm I} \right) = {\rm{Re}}\left( {{{\boldsymbol{\alpha }}_{\rm I}^H}{\bf{d}}_{\rm I}} \right) \mathop \gtrless\limits_{{\cal H}_0}^{{\cal H}_1} {\delta}^{'},
\end{align}
where $\delta^{'}$ denotes the threshold related to $T\left( {\bf{d}} _{\rm I}\right)$.

%to derive the distribution of the detector.

%Then we can construct a desired

%Also,  Omit the subscript about tow Scenarios that mentioned before, the processed signal $\bf d$ varies under such two hypothesis, which is given by
%\begin{align}
%\left\{ {\begin{array}{*{20}{c}}
%{{\cal H}_1}: {{\bf{d}}} = {{\boldsymbol{\alpha}}} + {{\hat {\bf z}}},\\
%{{\cal H}_0}: {{\bf{d}}} =  {{\hat {\bf z}}}.
%\end{array}} \right.
%\end{align}

%Next, we apply the Neyman-Pearson (NP) Lemma to obtain the LRT detector, and analysis its distribution for further derivation about false alarm probability and detection probability \cite{sssssss}. The likelihood ratio $L\left( {\bf{d}} \right)$ is given by

%where $\delta$ denotes the threshold. The simplified formula is further obtained by
%\begin{align}
%\ln L\left( {\bf{d}} \right) = \frac{1}{{{\sigma_d ^2}}}\left( {2{\rm{Re}}\left( {{{\boldsymbol{\alpha }}^H}{\bf{d}}} \right) - {{\boldsymbol{\alpha }}^H}{\boldsymbol{\alpha }}} \right), \label{mod:3}
%\end{align}
%where $\boldsymbol{\alpha }$ is the desired signal under the case that the target exist.

%Thus, we can compute and judge the case that ${\cal H}_1$ is hold by the first term in (\ref{mod:3}) and obtain the following detector

Then, we derive the distribution of $T\left( {\bf{d}}_{\rm I} \right)$. Towards this end, we consider $x={{{\boldsymbol{\alpha }}_{\rm I}^H}{\bf{d}}_{\rm I}}$, %= \sum\limits_{q = 1}^{N_r{K^2}} {\alpha _q^ * {d_q}} $.
whose expectation and variance under hypothesis ${{\cal H}_1}$ and ${{\cal H}_0}$ are respectively obtained as
\begin{align}
{\mathbb{E} }\left( {x|{{\cal H}_0}} \right) %= \sum\limits_{q = 1}^{N_rK^{2}} {{\mathbb{E} }\left( {{d_q}} \right){\alpha _q^*}}
& = 0, \label{exph0} \\
{\mathbb{E} }\left( {x|{\cal H}_1} \right) & = {\cal E} _{\rm I} \triangleq \sum\limits_{l \in {\cal L}} \sum\limits_{m \in {\cal L}} ( \sum\limits_{i \in {\cal K}_l} {\left\| {\bf{H}}_{m,l}{\bf{w}}_{l,i} \right\|}^2 \nonumber \\  &+ \sum\limits_{k = 1}^{r_t} {\left\| {\bf{H}}_{m,l}{\bf{w}}_{l,k}^r \right\|}^2  )   \nonumber\\ &=N_r\sum\limits_{l \in {\cal L}} \sum\limits_{m \in {\cal L}} \zeta_{m,l}^2{\beta _{m,l}}( \sum\limits_{i \in {\cal K}_l} {\left| {\bf{a}}_{t,l}^T\left( {{\theta _l}} \right){\bf{w}}_{l,i} \right|}^2 \nonumber \\  &+ {\bf{a}}_{t,l}^T\left( {{\theta _l}} \right){\bf{R}}_l^r{{\bf{a}}_{t,l}}\left( {{\theta _l}} \right)  )  , \label{exph1} \\
{\mathop{\rm var}} \left( {x|{\cal H}_0} \right) &= {\mathop{\rm var}} \left( {x|{\cal H}_1} \right)= {\sigma_d ^2} {\cal E}_{\rm I}. \label{varh0}
\end{align}
By combining (\ref{exph0}), (\ref{exph1}), and (\ref{varh0}), we have
\begin{align}
 \left\{ {\begin{array}{*{20}{c}}
x \sim{{\cal CN}\left( {0,{\sigma_d ^2}{\cal E}_{\rm I} } \right),{{\cal H}_0}},\\
x \sim{{\cal CN}\left( {{\cal E}_{\rm I},{\sigma_d ^2}{\cal E}_{\rm I} } \right),{{\cal H}_1}}.
\end{array}} \right.
\end{align}
As a result, for $T\left( {\bf{d}}_{\rm I} \right) = {\mathop{\rm Re}\nolimits} \left( x\right)$, it follows that
\begin{align}\label{distribution}
\left\{ {\begin{array}{*{20}{c}}
T\left( {\bf{d}}_{\rm I} \right) \sim {{\cal N}\left( {0,{{{{\sigma ^2_d}{\cal E}_{\rm I} } \mathord{\left/
 {\vphantom {{{\sigma ^2_d}} 2}} \right.
 \kern-\nulldelimiterspace} 2}}} \right),{{\cal H}_0}},\\
T\left( {\bf{d}}_{\rm I} \right) \sim {{\cal N}\left( {{\cal E}_{\rm I} ,{{{{\sigma ^2_d}{\cal E}_{\rm I} } \mathord{\left/
 {\vphantom {{{\sigma ^2_d}} 2}} \right.
 \kern-\nulldelimiterspace} 2}}} \right),{{\cal H}_1}}.
\end{array}} \right.
\end{align}

Finally, we derive the detection probability under a required false alarm probability. Based on \eqref{detector} and \eqref{distribution}, we obtain the detection probability $p_{\mathrm D}^{\rm I}$ and the false alarm probability ${p^{\rm I}_{\mathrm {FA}}}$ with respect to the detector threshold $\delta^{'}$ as
\begin{align}
{p^{\rm I}_{\mathrm D}} &= Q\left( {(\delta^{'}-{\cal E}_{\rm I} ) \sqrt {\frac{2}{{\sigma _d^2{\cal E}_{\rm I} }}} } \right),\label{pD1} \\
{p^{\rm I}_{\mathrm {FA}}} &= Q\left( {\delta^{'} \sqrt {\frac{2}{{\sigma _d^2{\cal E}_{\rm I} }}} } \right),\label{pFA1}
\end{align}
respectively.
Based on \eqref{pFA1}, we have ${\delta^{'} \sqrt {\frac{2}{{\sigma _d^2{\cal E}_{\rm I} }}} } ={Q^{ - 1}}\left( {{p^{\rm I}_{\mathrm {FA}}}} \right)$. By substituting this into \eqref{pD1}, we obtain the detection probability for given false alarm probability $p_{\mathrm {FA}}^{\rm I}$ as
\begin{align}
{p_{\mathrm D}^{\rm I}} =  Q\left( {{Q^{ - 1}}\left( {{p_{\mathrm {FA}}^{\rm I}}} \right) - \sqrt {\frac{{2{\cal E}_{\rm I} }}{{\sigma _d^2}}} } \right). \label{pdI}
\end{align}
It is observed that the detection probability ${p_{\mathrm D}^{\rm I}}$ in \eqref{pdI} is monotonically increasing with respect to ${\cal E}_{\rm I}$ in \eqref{exph1},
%\begin{align}
% {\cal E}_{\rm I} = N_r\sum\limits_{l \in {\cal L}} \sum\limits_{m \in {\cal L}} \zeta_{m,l}^2{\beta _{m,l}}\left( \sum\limits_{i \in {\cal K}} {{{\left| {\bf{a}}_{t,l}^T\left( {{\theta _l}} \right){\bf{w}}_{l,i} \right|}^2}}  + {\bf{a}}_{t,l}^T\left( {{\theta _l}} \right){\bf{R}}_l^r{{\bf{a}}_{t,l}}\left( {{\theta _l}} \right)  \right),\label{energyI1}
%\end{align}
which corresponds to the total received reflection-signal power over both direct and cross reflection links. Denote ${\bf{A}}_l({\theta _l}) = {\bf{a}}_{t,l}^ * ({\theta _l}){\bf{a}}_{t,l}^T({\theta _l}),\forall l\in\mathcal{L}$, ${\cal E}_{\rm I}$ is reexpressed  as
\begin{align}
{\cal E}_{\rm{I}} &= N_r\sum\limits_{l \in {\cal L}} \sum\limits_{m \in {\cal L}} \zeta_{m,l}^2{\beta _{m,l}}( \sum\limits_{i \in {\cal K}_l} {\rm tr}\left( {\bf{w}}_{l,i}{\bf{w}}_{l,i}^H{{\bf A}_l}\left( \theta _l \right) \right) \nonumber \\  &+ {\rm tr}\left( {\bf{R}}_l^r{{\bf A}_l}\left( \theta _l \right) \right) ). \label{energyI}
\end{align}

%It follows from \eqref{pdI} that the detection probability $p_{\mathrm{D}}^{\rm I}$ is monotonically increasing with respect to the total reflection-signal power ${\cal E}_{\rm{I}}$.
As a result, maximizing the detection performance of the system is equivalent to maximizing the received reflection-signal power ${\cal E}_{\rm{I}}$ in \eqref{energyI}.

\subsection{Detection Probability in Scenario \uppercase\expandafter{\romannumeral2} without BSs Synchronization}\label{Sec:III-B}
Next, we consider Scenario II without synchronization among the BSs. The detection probability in this scenario can be similarly derived as that in Scenario I, by replacing ${\bf{d}}_{\rm I}$ by ${\bf{d}}_{\rm II}$ and accordingly replacing ${\cal E}_{\rm I}$ in \eqref{exph1} by
\begin{align}
{\cal E}_{\rm II} &= N_r \sum\limits_{m \in {\cal L}} \zeta_{m,m}^2{\beta _{m,m}}( \sum\limits_{i \in {\cal K}_m} {{{\left| {\bf{a}}_{t,m}^T\left( \theta _m \right){\bf{w}}_{m,i} \right|}^2}} \nonumber \\ &+ {\bf{a}}_{t,m}^T\left( {{\theta _l}} \right){\bf{R}}_m^r{{\bf{a}}_{t,m}}\left( {{\theta _m}} \right)  )\nonumber \\
&= {N_r}\sum\limits_{m \in {\cal L}} \zeta_{m,m}^2{\beta _{m,m}}( \sum\limits_{i \in {\cal K}_m} {{\rm{tr}}} \left( {{{\bf{w}}_{m,i}}{\bf{w}}_{m,i}^H{{\bf{A}}_m}\left( {\theta _m} \right)} \right) \nonumber \\ &+ {\rm{tr}}\left( {{\bf{R}}_m^r{{\bf{A}}_m}\left( {\theta _m} \right)} \right) ), \label{energyII}
\end{align}
where ${{\bf{A}}_m}\left( {\theta _m} \right) = {\bf{a}}_{t,m}^*(\theta _m){\bf{a}}_{t,m}^T(\theta _m),\forall m \in {\cal L}$.

Based on the similar derivation procedure as in Section \ref{Sec:III-A}, we have the detection probability ${p_{\mathrm D}^{\rm II}}$ for a given false alarm probability ${p_{\mathrm {FA}}^{\rm II}}$ as
\begin{align}
{p_{\mathrm D}^{\rm II}} = Q\left( {{Q^{ - 1}}\left( {p_{\mathrm {FA}}^{\rm II}} \right) - \sqrt {\frac{{2{\cal {E}_{\rm II}}}}{{\sigma _d^2}}} } \right). \label{pdII}
\end{align}

It is observed from \eqref{pdII} that the detection probability ${p_{\mathrm D}^{\rm II}}$ is monotonically increasing with respect to the total received reflection-signal power over the direct links only, i.e.,  ${\cal E}_{\rm II}$ in \eqref{energyII}. Therefore, maximizing ${p_{\mathrm D}^{\rm II}}$ is equivalent to maximizing the received reflection-signal power ${\cal E}_{\rm II}$ in \eqref{pdII}.

\section{SINR-constrained Detection Probability Maximization via Coordinated Transmit Beamforming}
In this section, we design the coordinated transmit beamforming $\{{\bf{w}}_{l,k}\}$ and dedicated sensing signal covariance matrix $\{{\bf R}_{l}^r\}$ to maximize the minimum detection probability with a given false alarm probability $p_{\mathrm {FA}}$ over a particular targeted area, subject to the minimum SINR requirement $\Gamma_{m,k}$ at each cell $m\in \cal L$ and CU $k\in {\cal K}_m$, and the maximum power constraint $P_{\max}$ at each BS. In particular, let $\cal Q$ denote the targeted area for detection. To facilitate the design, we select $Q$ sample locations from $\cal Q$, denoted by $(x_0^{(q)}, y_0^{(q)}), \forall q\in {\cal Q} \triangleq \{1,\ldots, Q\}$. For a potential target located at $(x_0^{(q)}, y_0^{(q)})$, we denote the target angle with respect to BS $l\in\mathcal L$ as $\theta_l^{(q)}$ and the round-trip path loss from BS $l\in\mathcal L$ to target to BS $m\in\mathcal L$ as ${\beta _{m,l}^{(q)}}$.

%In both Scenario \rm{I} with time synchronization and Scenario \rm{II} without time synchronization, we are interested in designing the coordinated transmit beamforming $\{{\bf{w}}_{l,k}\}$ and dedicated sensing signal covariance matrix $\{{\bf R}_{l}^r\}$ for Type-I and Type-II CU receivers to maximize the minimum detection probability over a particular target area with a given false alarm probability $p_{FA}$, subject to the minimum SINR requirement $\Gamma_{m,k}$ at each cell $m\in \cal L$ and CU $k\in \cal K$, and the maximum power constraint $P_{\max}$ at each BS.

\subsection{Scenario \uppercase\expandafter{\romannumeral1} with BSs Synchronization}

First, we consider Scenario I with BSs synchronization. In this scenario, for given target location $q$, maximizing the detection probability in this scenario is equivalent to maximizing ${\cal E}_{\rm I}$ in \eqref{energyI}. Based on this observation, the SINR-constrained minimum detection probability  maximization problems over the given targeted area with Type-I and Type-II receivers are formulated as
\begin{subequations}
\begin{align}
\left( {\rm{P1}} \right): \mathop {\max }\limits_{\left\{ {\bf{w}}_{l,i}, {{\bf{R}}_l^r} \right\}} ~ &~ \mathop {\min }\limits_{q\in \cal{Q}}  ~{f^{\rm{I}}}\left( {\{ {{\bf{w}}_{l,i}}\} ,\{ {\bf{R}}_l^r\} } \right)    \label{prob:1} \\
~{\rm{s.t.}}&~~\gamma _{m,k}^{\rm I}\left( {\left\{ {{{\bf{w}}_{l,i}}} \right\},\left\{ {{\bf{R}}_l^r} \right\}} \right)\ge {\Gamma _{m,k}}, \nonumber \\  & ~~\forall k \in {\cal K}_m, \forall m \in {\cal L}, \label{prob:con:1} \\
&~~\sum\limits_{i \in {\cal K}_l} {\left\| {{\bf{w}}_{l,i}} \right\|}^2 + {\rm Tr} ({{\bf R}_{l}^r})  \le P_{\max },\forall l \in {\cal L}, \label{prob:con:2}\\
&~~{\bf{R}}_l^r \succeq {\bf 0}, \forall l \in {\cal L},\label{prob:con:3}
\end{align}
\end{subequations}
and
\begin{subequations}
\begin{align}
\left( {\rm{P2}} \right): \mathop {\max }\limits_{\left\{ {\bf{w}}_{l,i}, {{\bf{R}}_l^r} \right\}} ~ &~ \mathop {\min }\limits_{q\in \cal{Q}}  ~{f^{\rm{I}}}\left( {\{ {{\bf{w}}_{l,i}}\} ,\{ {\bf{R}}_l^r\} } \right)   \label{prob:2} \\
~{\rm{s.t.}}&~~ \gamma _{m,k}^{\rm II}\left( {\left\{ {{{\bf{w}}_{l,i}}} \right\},\left\{ {{\bf{R}}_l^r} \right\}} \right)\ge {\Gamma _{m,k}}, \nonumber \\ & ~~\forall k \in {\cal K}_m, \forall m\in {\cal L}, \label{prob22:con:1} \\
&~~\eqref{prob:con:2}~{\text {and}}~\eqref{prob:con:3},\nonumber
\end{align}
\end{subequations}
%and
%\begin{subequations}
%\begin{align}
%\left( {\rm{P1-II}} \right): \mathop {\max }\limits_{\left\{ {\bf{w}}_{l,i}, {{\bf{R}}_l^r} \right\}} ~ &~ \mathop {\min }\limits_{q\in \cal{Q}}  ~\sum\limits_{l \in {\cal L}} {\sum\limits_{m \in {\cal L}} {{\beta _{m,l}}{\rm tr}\left( {\left( {\sum\limits_{i \in K} {\bf{w}}_{l,i}{\bf{w}}_{l,i}^H  + {\bf{R}}_l^r} \right){{\bf{A}}_l}\left( \theta _l^{(q)} \right)} \right)} }    \label{probII:1} \\
%~{\rm{s.t.}}&~~\gamma _{m,k}^{\rm II}\{ {\bf w}_{l,i}, {\bf R}_l^r\}\ge {\Gamma _{m,k}}, \forall k \in {\cal K}, m\in {\cal L} \label{probII:con:1} \\
%&~~\eqref{prob:con:2}~{\text{and}}~\eqref{prob:con:3},
%\end{align}
%\end{subequations}
respectively, where
\begin{align}
{f^{\rm{I}}}\left( {\{ {{\bf{w}}_{l,i}}\} ,\{ {\bf{R}}_l^r\} } \right) &= \sum\limits_{l \in {\cal L}} \sum\limits_{m \in {\cal L}} \zeta _{m,l}^2\beta _{m,l}^{(q)}{\rm{tr}}( ( \sum\limits_{i \in {\cal K}_l} {{\bf{w}}_{l,i}} {\bf{w}}_{l,i}^H \nonumber \\&+ {\bf{R}}_l^r ){{\bf{A}}_l}( {\theta _l^{(q)}} ) ) ,
\end{align}
\eqref{prob:con:1} and \eqref{prob22:con:1} denote the minimum SINR constraints at different types of CUs and \eqref{prob:con:2} denotes the maximum transmit power constraints at the BSs. Notice that problems (P1) and (P2) are non-convex due to the non-convex constraints in \eqref{prob:con:1} and \eqref{prob22:con:1}. In the following, we apply the SDR technique to solve problems (P1) and (P2).

%We assume that the CC and CUs perfectly know the channel state information (CSI) of all path, i.e. ${{\bf{h}}_{k,i}}, \forall i,k \in {\cal K}$. This assumption is made to gain some design insights by analyze the upper bound of performance while the practical system may not perfectly know the CSI.

Towards this end, we introduce $\Omega $ as an auxiliary optimization variable and define ${{\bf{W}}_{l,i}} ={\bf{w}}_{l,i}{\bf{w}}_{l,i}^H \succeq \bf{0}$ with ${\rm {rank}} ({\bf{W}}_{l,i})\le 1$, $\forall l \in {\cal L}, i \in {\cal K}_l$. Problems (P1) and (P2) are equivalently reformulated as
\begin{subequations}
\begin{align}
\left( {\rm{P1.1}} \right):&\mathop {\max }\limits_{\left\{ {{{\bf{W}}_{l,i}}\succeq \bf{0}},{\bf{R}}_l^r \succeq {\bf{0}},\Omega \right\}} ~\Omega    \label{prob:1.1} \\
&~~~~{\rm{s.t.}}~ {{\hat f}^{\rm{I}}}\left( {\{ {{\bf{W}}_{l,i}}\} ,\{ {\bf{R}}_l^r\} } \right) \ge \Omega, \forall q \in {\cal Q}, \label{prob1.1:t} \\
&~~~~~~~~\sum\limits_{l \in {\cal L}} {\sum\limits_{i \in {\cal K}_l} {{\rm{tr}}\left( {{{\bf{h}}_{l,m,k}}{\bf{h}}_{l,m,k}^H{{\bf{W}}_{l,i}}} \right)} } \nonumber \\
&~~~~~~~~+ \sum\limits_{l \in {\cal L}} {\rm tr}\left( {{{\bf{h}}_{l,m,k}}{\bf{h}}_{l,m,k}^H{\bf{R}}_l^r} \right)\nonumber \\
&~~~~~~~~+  \sigma _c^2  \le ( {1 + \frac{1}{{{\Gamma _{m,k}}}}} ){\rm{tr}}\left( {{{\bf{h}}_{m,m,k}}{\bf{h}}_{m,m,k}^H{{\bf{W}}_{m,k}}} \right), \nonumber \\
&~~~~~~~~~~\forall k \in {\cal K}_m, \forall m\in {\cal L}, \label{prob1.1:con:1} \\
&~~~~~~~~\sum\limits_{i \in {\cal K}_l} {\rm tr}\left( {\bf{W}}_{l,i}\right)  + {\rm tr} ({{\bf R}_{l}^r})  \le P_{\max },\forall l \in {\cal L}, \label{prob1.1:con:2}\\
&~~~~~~~~{\rm {rank}} ({\bf{W}}_{l,i}) \le 1, \forall i \in {\cal K}_l, \forall l \in {\cal L} \label{prob1.1:con:3}
\end{align}
\end{subequations}
and
\begin{subequations}
\begin{align}
\left( {\rm{P2.1}} \right):&\mathop {\max }\limits_{\left\{ {{{\bf{W}}_{l,i}}\succeq \bf{0}},{\bf{R}}_l^r \succeq {\bf{0}},\Omega \right\}} ~\Omega   \\
&~~~~{\rm{s.t.}}~ \sum\limits_{l \in {\cal L}} {\sum\limits_{i \in {\cal K}_l} {{\rm{tr}}\left( {{{\bf{h}}_{l,m,k}}{\bf{h}}_{l,m,k}^H{{\bf{W}}_{l,i}}} \right)} } +  \sigma _c^2 \nonumber \\
&~~~~~~~~\le \left( {1 + \frac{1}{{{\Gamma _{m,k}}}}} \right){\rm{tr}}\left( {{{\bf{h}}_{m,m,k}}{\bf{h}}_{m,m,k}^H{{\bf{W}}_{m,k}}} \right), \nonumber \\
&~~~~~~~~~~\forall k \in {\cal K}_m, \forall  m \in {\cal L}, \label{prob22.1:con:2}  \\
&~~~~~~~~~~\eqref{prob1.1:t},~\eqref{prob1.1:con:2},~\text{and}~\eqref{prob1.1:con:3}, \nonumber
\end{align}
\end{subequations}
respectively, where
\begin{align}
{{\hat f}^{\rm{I}}}\left( {\{ {{\bf{W}}_{l,i}}\} ,\{ {\bf{R}}_l^r\} } \right) &= \sum\limits_{l \in {\cal L}} \sum\limits_{m \in {\cal L}} \zeta _{m,l}^2\beta _{m,l}^{(q)}{\rm{tr}}( ( \sum\limits_{i \in {\cal K}_l} {{{\bf{W}}_{l,i}}} \nonumber \\ & + {\bf{R}}_l^r ){{\bf{A}}_l}( {\theta _l^{(q)}} ) )
\end{align}
%\begin{subequations}
%\begin{align}
%\left( {\rm{P1.1-II}} \right):&\mathop {\max }\limits_{\left\{ {{{\bf{W}}_{l,i}}\succeq \bf{0}},{{\bf{R}}_l^r},t \right\}} ~t    \label{probII:1.1} \\
%&~~~~{\rm{s.t.}}~ \sum\limits_{l \in {\cal L}} {\sum\limits_{i \in {\cal K}} {{\rm{tr}}\left( {{{\bf{h}}_{l,m,k}}{\bf{h}}_{l,m,k}^H{{\bf{W}}_{l,i}}} \right)} }  + \sum\limits_{l \in {\cal L}} {\rm tr}\left( {{{\bf{h}}_{l,m,k}}{\bf{h}}_{l,m,k}^H{\bf{R}}_l^r} \right)\nonumber \\
%&~~~~+ \sigma _c^2  \le \left( {1 + \frac{1}{{{\Gamma _{m,k}}}}} \right){\rm{tr}}\left( {{{\bf{h}}_{m,m,k}}{\bf{h}}_{m,m,k}^H{{\bf{W}}_{l,k}}} \right), \forall k \in {\cal K} \label{probII:con:1.1} \\
%&~~~~~~~~\eqref{prob:t},~\eqref{prob:con:2.1},{\text{ and}}~\eqref{prob:con:3.1}
%\end{align}
%\end{subequations}

However, problems (P1.1) and (P2.1) are still non-convex due to the rank-one constraints in \eqref{prob1.1:con:3}. To tackle this issue, we drop these rank-one constraints and obtain the SDR version of (P1.1) and (P2.1) as (SDR1.1) and (SDR2.1) \cite{luo2010semidefinite} respectively, both of them are convex and can be optimally solved by standard convex optimization program solvers such as CVX \cite{grant2014cvx}. Let $\{ \{ {\bf{W}}_{l,i}^ * \} ,\{ {\bf{R}}_l^{r * }\}, \Omega  ^ *\}$ and $\{ \{ {\bf{W}}_{l,i}^{**}\} ,\{ {\bf{R}}_l^{r**}\} ,{\Omega ^{**}}\} $ denote the optimal solutions to (SDR1.1) and (SDR2.1), respectively. Notice that the obtained ${\bf{W}}_{l,i}^*$  and ${\bf{W}}_{l,i}^{**}$ are generally of high ranks, which do not necessarily satisfy the rank-one constraints in (P1.1) and (P2.1). As such, we introduce the following additional step to construct the equivalent optimal rank-one solutions to (P1.1) and (P2.1).

\begin{prop}
The SDR of problem (P1.1) is tight. In particular, based on the optimal solution of $\{ \{ {\bf{W}}_{l,i}^ * \} ,\{ {\bf{R}}_l^{r * }\}, \Omega  ^ *\}$ to (SDR1.1), if any of $ \{ {\bf{W}}_{l,i}^ * \}$ is not rank-one, we can always construct the equivalent optimal rank-one solution of $\{ \{ {{{\bf{\tilde W}}}_{l,i}}\} ,\{ {\bf{\tilde R}}_l^r\} ,\tilde \Omega \} $ to (P1.1) according to the following, which achieves the same objective value as (SDR1.1):
\begin{subequations}
\begin{align}
{{\bf{\tilde w}}_{l,i}} &= {\left( {{\bf{h}}_{l,m,k}^H{\bf{W}}_{l,i}^ * {\bf{h}}_{l,m,k}} \right)^{-\frac{1}{2}}}{\bf{W}}_{l,i}^ * {\bf{h}}_{l,m,k}, \label{app:1:1}  \\
{{{\bf{\tilde W}}}_{l,i}} &= {{\bf{\tilde w}}_{l,i}}{\bf{\tilde w}}_{l,i}^H, \label{app:1:2} \\
{\bf{\tilde R}}_l^r &= \sum\limits_{i \in {\cal K}_l} {{\bf{W}}_{l,i}^*}   + {\bf{\tilde R}}_l^{r * } - \sum\limits_{i \in {\cal K}_l} {{\bf{\tilde W}}_{l,i}^*} \label{app:1:3}\\
\tilde \Omega &={\Omega ^*} . \label{app:1:4}
\end{align}
\end{subequations}
\end{prop}

\begin{proof}
See Appendix \ref{appendixA}.
\end{proof}

Similarly as for problem (P1.1), we find the optimal solution to (P2.1) by showing that the SDR is tight in the following proposition.

\begin{prop}
The SDR of problem (P2.1) is tight. In particular, based on the optimal solution of $\{ \{ {\bf{W}}_{l,i}^ {**} \} ,\{ {\bf{R}}_l^{r ** }\}, \Omega  ^ {**}\}$ to (SDR2.1), if any of $ \{ {\bf{W}}_{l,i}^ {**} \}$ is not rank-one, we can always construct the equivalent optimal rank-one solution of $\{ \{ {{{\bf{\bar W}}}_{l,i}}\} ,\{ {\bf{\bar R}}_l^r\} ,\bar \Omega \}$ to (P2.1) in the following, which achieves the same objective value as (SDR2.1):
\begin{subequations}
\begin{align}
{{\bf{\bar w}}_{l,i}} &= {\left( {{\bf{h}}_{l,m,k}^H{\bf{W}}_{l,i}^{**}{{\bf{h}}_{l,m,k}}} \right)^{ - \frac{1}{2}}}{\bf{W}}_{l,i}^{**}{{\bf{h}}_{l,m,k}}, \label{app:II:1}  \\
{{{\bf{\bar W}}}_{l,i}} &= {{\bf{\bar w}}_{l,i}}{\bf{\bar w}}_{l,i}^H, \label{app:II:2} \\
{\bf{\bar R}}_l^r &= \sum\limits_{i \in {\cal K}_l} {{\bf{W}}_{l,i}^{**}}  + {\bf{R}}_l^{r**} - \sum\limits_{i \in {\cal K}_l} {{{{\bf{\bar W}}}_{l,i}}}  \label{app:II:3}\\
\bar \Omega &={\Omega ^{**}}. \label{app:II:4}
\end{align}
\end{subequations}
\end{prop}

\begin{proof}
The proof is similar to that in Appendix \ref{appendixA}, for which the details are omitted.
\end{proof}

\subsection{Scenario \uppercase\expandafter{\romannumeral2} without BSs Synchronization }

Next, we consider Scenario II without BSs synchronization. In this scenario, the SINR-constrained minimum detection probability maximization problems of Type-I and Type-II CU receivers are respectively formulated as problems (P3) and (P4) in the following, which are similar to problems (P1) and (P2) by replacing $\cal{E}_{\rm{I}}$ in \eqref{energyI} as $\cal{E}_{\rm{II}}$ in \eqref{energyII}, respectively:
\begin{align}
\left( {\rm{P3}} \right): \mathop {\max }\limits_{\left\{ {\bf{w}}_{m,i}, {\bf{R}}_m^r\succeq \bf{0} \right\}} ~ &~ \mathop {\min }\limits_{q\in \cal{Q}}  ~ {f^{{\rm{II}}}}\left( {\{ {{\bf{w}}_{m,i}}\} ,\{ {\bf{R}}_m^r\} } \right)  \\
~{\rm{s}}{\rm{.t}}{\rm{.}}&~~\eqref{prob:con:1}~{\text{and}}~\eqref{prob:con:2}. \nonumber
\end{align}

\begin{align}
\left( {\rm{P4}} \right): \mathop {\max }\limits_{\left\{ {\bf{w}}_{m,i},{{\bf{R}}_m^r\succeq \bf{0}} \right\}} ~ &~ \mathop {\min }\limits_{q\in \cal{Q}}  ~ {f^{{\rm{II}}}}\left( {\{ {{\bf{w}}_{m,i}}\} ,\{ {\bf{R}}_m^r\} } \right)  \\
~{\rm{s}}{\rm{.t}}{\rm{.}}&~~\eqref{prob:con:2}~{\text{and}}~\eqref{prob22:con:1}. \nonumber
\end{align}
where
\begin{align}
{f^{{\rm{II}}}}\left( {\{ {{\bf{w}}_{m,i}}\} ,\{ {\bf{R}}_m^r\} } \right) &= \sum\limits_{m \in {\cal L}} \zeta _{m,m}^2\beta _{m,m}^{(q)}{\rm{tr}}( ( \sum\limits_{i \in {\cal K}_m} {{{\bf{w}}_{m,i}}} {\bf{w}}_{m,i}^H \nonumber \\
&+ {\bf{R}}_m^r ){{\bf{A}}_m}( {\theta _m^{(q)}} ) ).
\end{align}
%and
%\begin{align*}
%\left( {\rm{P2-II}} \right): \mathop {\max }\limits_{\left\{ {\bf{w}}_{l,i},{{\bf{R}}_l^r} \right\}} ~ &~ \mathop {\min }\limits_{q\in \cal{Q}}  ~ {\sum\limits_{m \in {\cal L}} {{\beta _{m,m}}{\rm tr}\left( {\left( {\sum\limits_{i \in {\cal K}} {\bf{w}}_{m,i}{\bf{w}}_{m,i}^H  + {\bf{R}}_m^r} \right){{\bf{A}}_m}\left( \theta _m^{(q)} \right)} \right)} }  \\
%~{\rm{s}}{\rm{.t}}{\rm{.}}&~~\eqref{probII:con:1},~\eqref{prob:con:2},~{\text{and}}~\eqref{prob:con:3}.
%\end{align*}

As problems (P3) and (P4) have similar structures as problems (P1) and (P2), respectively, they can also be solved optimally based on the SDR. More specifically, by introducing the auxiliary variable $\Omega$, and defining ${{\bf{W}}_{m,i}} ={\bf{w}}_{m,i}{\bf{w}}_{m,i}^H \succeq \bf{0}$ with ${\rm {rank}} ({\bf{W}}_{m,i})\le 1$, problems (P3.1) and (P4.1) can be reformulated equivalently as

\begin{subequations}
\begin{align}
\left( {\rm{P3.1}} \right):&\mathop {\max }\limits_{\left\{ {{{\bf{W}}_{m,i}}\succeq \bf{0}},{\bf{R}}_m^r\succeq {\bf{0}},\Omega \right\}} ~\Omega    \label{prob:4.1} \\
&~~~~~~~~{\rm{s.t.}}~  {{\hat f}^{{\rm{II}}}}\left( {\{ {{\bf{W}}_{m,i}}\} ,\{ {\bf{R}}_m^r\} } \right)  \ge \Omega, \forall q \in {\cal Q}, \label{prob4.1:t} \\
%&~~~~~~~~\sum\limits_{l \in {\cal L}} {\sum\limits_{i \in {\cal K}} {{\rm{tr}}\left( {{{\bf{h}}_{l,m,k}}{\bf{h}}_{l,m,k}^H{{\bf{W}}_{l,i}}} \right)} }  + \sum\limits_{l \in {\cal L}} {\rm tr}\left( {{{\bf{h}}_{l,m,k}}{\bf{h}}_{l,m,k}^H{\bf{R}}_l^r} \right)\nonumber \\
%&~~~~~~~~+  \sigma _c^2  \le \left( {1 + \frac{1}{{{\Gamma _{m,k}}}}} \right){\rm{tr}}\left( {{{\bf{h}}_{m,m,k}}{\bf{h}}_{m,m,k}^H{{\bf{W}}_{m,k}}} \right), \forall k \in {\cal K}, m\in {\cal L} \label{prob3.1:con:1} \\
%&~~~~~~~~\sum\limits_{i \in {\cal K}} {\rm tr}\left( {\bf{W}}_{l,i}\right)  + {\rm tr} ({{\bf R}_{l}^r})  \le P_{\max },\forall l \in {\cal L} \label{prob3.1:con:2}\\
%&~~~~~~~~{\rm {rank}} ({\bf{W}}_i)=1, \forall i \in \cal K.\label{prob3.1:con:3}
&~~~~~~~~\eqref{prob1.1:con:1}, ~\eqref{prob1.1:con:2},~\text{and}~\eqref{prob1.1:con:3},\nonumber
\end{align}
\end{subequations}
and
\begin{subequations}
\begin{align}
\left( {\rm{P4.1}} \right):&\mathop {\max }\limits_{\left\{ {{{\bf{W}}_{m,i}}\succeq \bf{0}},{\bf{R}}_m^r \succeq {\bf{0}},\Omega \right\}} ~\Omega    \\
%&~~~~{\rm{s.t.}}~  {\sum\limits_{m \in L} {{\beta _{m,m}}{\rm tr}\left( {\left( {\sum\limits_{i \in {\cal K}} {\bf{W}}_{m,i}  + {\bf{R}}_m^r} \right){{\bf{A}}_m}\left( \theta _m^{(q)} \right)} \right)} }  \ge t, \forall q \in {\cal Q}  \\
%&~~~~~~~~\sum\limits_{l \in {\cal L}} {\sum\limits_{i \in {\cal K}} {{\rm{tr}}\left( {{{\bf{h}}_{l,m,k}}{\bf{h}}_{l,m,k}^H{{\bf{W}}_{l,i}}} \right)} } +  \sigma _c^2 \nonumber \\
%&~~~~~~~~\le \left( {1 + \frac{1}{{{\Gamma _{m,k}}}}} \right){\rm{tr}}\left( {{{\bf{h}}_{m,m,k}}{\bf{h}}_{m,m,k}^H{{\bf{W}}_{l,k}}} \right), \forall k \in {\cal K}, m\in {\cal L}  \\
&~~~~~~~~{\rm{s.t.}}~\eqref{prob4.1:t},~\eqref{prob22.1:con:2},~ \eqref{prob1.1:con:2},~\text{and}~\eqref{prob1.1:con:3},\nonumber
\end{align}
\end{subequations}
respectively, where
\begin{align}
{{\hat f}^{{\rm{II}}}}\left( {\{ {{\bf{W}}_{m,i}}\} ,\{ {\bf{R}}_m^r\} } \right) &= \sum\limits_{m \in {\cal L}} \zeta _{m,m}^2\beta _{m,m}^{(q)}{\rm{tr}}( ( \sum\limits_{i \in {\cal K}_m} {{{\bf{W}}_{m,i}}} \nonumber \\  &+ {\bf{R}}_m^r ){{\bf{A}}_m}( {\theta _m^{(q)}} ) )
\end{align}

Then, we drop the rank-one constraints on $\{ {\bf W}_{m,i}\}$ in \eqref{prob1.1:con:3} to obtain the SDR versions of (P3.1) and (P4.1) as (SDR3.1) and (SDR4.1), respectively, which are convex and can be solved optimally.

%For problem (SDR3.1) and (SDR4.1), denote the optimal high-rank solution as $\{ \{ {\bf{W}}_{l,i}^ * \} ,\{ {\bf{R}}_l^{r * }\}, t \}$, $l \in {\cal L}, i \in {\cal K}$, then we can construct new solutions such that

%\begin{subequations}
%\begin{align}
%{{\bf{\bar w}}_{l,i}} = {\left( {{\bf{h}}_{l,m,k}^H{\bf{W}}_{l,i}^ * {\bf{h}}_{l,m,k}} \right)^{-\frac{1}{2}}}{\bf{W}}_{l,l}^ * {\bf{h}}_{l,m,k}, \label{app:II:1}  \\
%{{{\bf{\bar W}}}_{l,i}} = {{\bf{\bar w}}_{l,i}}{\bf{\bar w}}_{l,i}^H, \label{app:II:2} \\
%{\bf{\bar R}}_l^r = \sum\limits_{i \in {\cal K}} {{\bf{W}}_{l,i}^*}   + {\bf{\bar R}}_l^{r * } - \sum\limits_{i \in {\cal K}} {{\bf{\bar W}}_{l,i}^*} \label{app:II:3}\\
%\bar t=t. \label{app:II:4}
%\end{align}
%\end{subequations}

%After substitute the constructed solution, it is observed that the objective function in (SDR3.1) and (SDR4.1) satisfy $\bar t = t$. It follows that ${\bf{\bar R}}_l^r + \sum\limits_{i \in {\cal K}} {{{{\bf{\bar W}}}_{l,i}}}  = \sum\limits_{i \in {\cal K}} {{\bf{W}}_{l,i}^*}  + {\bf{\bar R}}_l^{r*}$ according to \eqref{app:II:3}, thus constraints \eqref{prob1.1:con:2} and  \eqref{prob4.1:t} are also satisfied. For the QoS constraints in (SDR3.1) and (SDR4.1), \eqref{prob1.1:con:1} and \eqref{prob22.1:con:2} also satisfied due to \eqref{qos:1.1} and \eqref{qos:2.1}. Thus, the relaxed problem (SDR3.1) and (SDR4.1) achieve the same optimality with constructed rank-one solution $\{ \{ {\bf{\bar W}}_{l,i} \} ,\{ {\bf{\bar R}}_l\}, \bar t \}$, $l \in {\cal L}, i \in {\cal K}$.

Note that problems (SDR3.1) and (SDR4.1) generally have high rank optimal solutions, which may not satisfy the rank-one constraints in \eqref{prob1.1:con:3} for (P3.1) and (P4.1). Fortunately, by following the similar concepts as in {\textit{Propositions 1}} and {\textit{Propositions 2}} for the SDRs of (P3.1) and (P4.1), one can show that the optimal rank-one solutions to (P3.1) and (P4.1) can always be constructed. The details of the derivations are thus omitted for brevity.

\begin{remark}
Comparing problem (P1.1) (or (P2.1)) in Scenario I with problem (P3.1) (or (P4.1)) in Scenario II, we observe that the feasible solutions of (P1.1) (or (P3.1)) is a subset of (P2.1) (or (P4.1)), but not vice versa. Thus, we conclude that problem (P2.1) (or (P4.1)) can always achieve a higher optimal objective value or at least equal to that of (P1.1) (or (P3.1)), since (P2.1) and (P4.1) enjoy a larger feasible solution set than (P1.1) and (P3.1), respectively. Moreover, via extensive simulations in the next section, we observe that for Type-I receivers without sensing signal interference cancellation, the optimal solutions to (P2.1) and (P4.1) satisfy ${\bf{ R}}_l^r={\bf 0}$, which shows that employing dedicated sensing signals is not necessary. Intuitively, this is because the dedicated sensing signals would introduce harmful interference for communications in this case.
\end{remark}

\section{Numerical Results}
%\color{blue}

In this section, we provide numerical results to validate the performance of our proposed coordinated transmit beamforming designs for the multi-antenna networked ISAC system.

\subsection{Benchmark Schemes}
First, we consider the following benchmark schemes for performance comparison.

\begin{itemize}
%\item {\bf Communication design}: First, we design the transmit beamforming vectors $\{{{{\bf{w}}_i}}\}$ to minimize the total transmit power $\sum_{i=1}^K \|{\bf{w}}_i\|^2$ while ensuring the SINR constraints in \eqref{prob:con:1}, for which the obtained transmit beamformers are denoted by $\{{\bf{\bar{w}}}_i\}$. Then, we scale them as ${\bf{\hat{w}}}_i = \hat{p} {\bf{\bar{w}}}_i, \forall i \in \mathcal{K}$, where $\hat{p} > 0$ is set as the maximum value provided that the maximum transmit power constraints in \eqref{prob:con:2} are ensured.
%\item {\bf Communication without dedicated radar signals}: In our proposed schemes, set ${\bf{R}}_l^r = {\bf{0}}, \forall l$, then we can compare our schemes with conventional transmission optimizations, subjected to the maximum power constraints and SINR constraints.
\item {\bf ISAC with ZF information beamforming}: In this scheme, we apply coordinated ZF beamforming \cite{behdad2022power, heath2018foundations}. Let ${{{\bf{\bar H}}}_{m,m,k}} = [{{\bf{h}}_{m,1,1}}, \ldots ,{{\bf{h}}_{m,m,k - 1}},{{\bf{h}}_{m,m,k + 1}}, \ldots {{\bf{h}}_{m,L,K}}]$, and ${{{\bf{\bar H}}}_{m,m,k}} = {{{\bf{\bar U}}}_{m,m,k}}{{\bf{\bar \Lambda }}_{m,m,k}}{\bf{\bar V}}_{ m,m,k}^H$ denotes the application of singular value decomposition (SVD) on ${{{\bf{\bar H}}}_{m,m,k}}$, where ${{\bf{\bar U}}_{m,m,k}} = [{\bf{\bar U}}_{ m,m,k}^{{\rm{\bar n\bar u\bar l\bar l}}}{\bf{\bar U}}_{ m,m,k}^{{\rm{null}}}]$ and ${\bf{\bar U}}_{m,m,k}^{{\rm{null}}} \in {{\mathbb C}^{{N_t} \times ({N_t} - {L^2}K + 1)}}$. The ZF transmit beamforming at BS $m$ for CU $k$ is designed as
    \begin{align}
    {\bf{w}}_{m,k}^{{\rm{ZF}}} &= \frac{{\sqrt {p_{m,k}^{{\rm{ZF}}}} {\bf{\bar U}}_{m,m,k}^{{\rm{null}}}{\bf{\bar U}}_{m,m,k}^{{\rm{null }}H}{{\bf{h}}_{m,m,k}}}}{{\left\| {{\bf{\bar U}}_{ m,m,k}^{{\rm{null}}}{\bf{\bar U}}_{m,m,k}^{{\rm{null }}H}{{\bf{h}}_{m,m,k}}} \right\|}}, \nonumber \\
    &~~\forall m \in {\cal L}, \forall k \in {\cal K}_m, \label{zfbf}
    \end{align}
      where ${p_{m,k}^{{\rm{ZF}}}}$ denotes the power for CU $k$ by BS  $m$, which is a variable to be optimized. Accordingly, the power constraint at each BS $m$ becomes$\sum\limits_{k \in {\cal K}} {p_{m,k}^{{\rm{ZF}}}}  \le {P_{\max }}$. By substituting ${\bf{w}}_{m,k}^{{\rm{ZF}}}$ in \eqref{zfbf} into problems (P1.1), (P2.1), (P3.1), and (P4.1), we obtain the corresponding power allocation problems (ZF1.1), (ZF2.1), (ZF3.1), and (ZF4.1), which can be optimized similarly as in Section IV for obtaining the optimal coordinated power control solutions.
\item {\bf Joint detection via dedicated sensing signals}: In this scheme, the BSs only employ the dedicated sensing signals for joint detection. The corresponding detection probabilities in Scenario I and Scenario II are respectively given by
    \begin{align}
    p_{\mathrm D}^{\rm{I}} = Q\left( {{Q^{ - 1}}\left( {p_{\mathrm {FA}}^{\rm{I}}} \right) - \sqrt {\frac{{2{{\hat {\cal E}}_{\rm{I}}}}}{{\sigma _d^2}}} } \right),\\
    p_{\mathrm D}^{{\rm{II}}} = Q\left( {{Q^{ - 1}}\left( {p_{\mathrm {FA}}^{{\rm{II}}}} \right) - \sqrt {\frac{{2{{\hat {\cal E}}_{{\rm{II}}}}}}{{\sigma _d^2}}} } \right),
    \end{align}
    where
      \begin{align}
      {{\hat {\cal E}}_{\rm{I}}} = {N_r}\sum\limits_{l \in {\cal L}} {\sum\limits_{m \in {\cal L}} {{\zeta_{m,l} ^2}{\beta _{m,l}}{\rm{tr}}( {{\bf{R}}_l^r} {{\bf{A}}_l}( {\theta _l^{(q)}} ) )} },
      \end{align}
      \begin{align}
      {{\hat {\cal E}}_{\rm{II}}} = {N_r}{\sum\limits_{m \in {\cal L}} {{\zeta_{m,m} ^2}{\beta _{m,m}}{\rm{tr}}( {{\bf{R}}_m^r} {{\bf{A}}_m}( {\theta _m^{(q)}} ) )} },
      \end{align}
      %Then the constraints \eqref{prob1.1:t} in (P1.1) and \eqref{prob2.1:con:t} in (P2.1) (or (P3.1) and (P4.1)) are substituted by $\sum\limits_{l \in {\cal L}} {\sum\limits_{m \in {\cal L}} {\zeta_{m,l}^2{\beta _{m,l}}{\rm tr}( { {\sum\limits_{i \in {\cal K}}  {\bf{R}}_l^r} {{\bf{A}}_l}( \theta _l^{(q)} )} )} }$ (or ${\sum\limits_{m \in {\cal L}} {\zeta_{m,m}^2{\beta _{m,m}}{\rm tr}( { {\sum\limits_{i \in {\cal K}}  {\bf{R}}_m^r} {{\bf{A}}_m}( \theta _m^{(q)} )} )} }$). Thus, we obtain (D1.1), (D2.1), (D3.1), and (D4.1) as the transformed joint detection optimization problems.
      denote the correspondingly received echo signals of the dedicated sensing signals.
      Accordingly, we optimize the joint beamforming by solving problems (P1)-(P4) via replacing ${{ {\cal E}}_{\rm{I}}}$ and ${{ {\cal E}}_{\rm{II}}}$ as ${{\hat {\cal E}}_{\rm{I}}}$ and ${{\hat {\cal E}}_{\rm{II}}}$, respectively.

\end{itemize}

%In the simulation, each BS is deployed with a ULA with half wavelength spacing between antennas, and Rayleigh fading is considered for communication. The noise powers are set as $\sigma^2_c = -84$ dBm and $\sigma^2_d = -102$ dBm after the MF processing with the ISAC transmission duration being $N=64$. The SINR constraints at CUs in each cell are set to be identical, i.e.,  $\Gamma_{m,k}=\Gamma$, $\forall k \in {\cal K}, \forall m \in {\cal M}$. Furthermore, suppose that there are $L=3$ BSs and each cell with $K=4$ CUs in the system. The coordinates of the three BSs are $\left( {5 \text{m},21\text{m}} \right)$, $\left({-65 \text{m},-6\text{m}}\right)$, and $\left( {68\text{m},0 \text{m}} \right)$, the CUs at each cell are randomly seted. The numbers of transmit and receive antennas at BSs are $N_t=N_r=N_a$. In addition, the target area is set as a square with area of $2 \times 2=4$ square meters and with the center at $\left( {0 ,-30\text{m}} \right)$, and there are $M=9$ sample locations uniformly distributed in the target area. The simulation topology is illustrated in Figure 2.

\subsection{Simulation Results }

In the simulation, we consider the networked ISAC scenario with $L = 3 $ BSs as shown in Fig. 2, where each BS serves one CU or multiple CUs. Each BS is deployed with a ULA with half a wavelength spacing between the antennas. The noise powers are set as $\sigma^2_c = -84$ dBm and $\sigma^2_d = -102$ dBm. The SINR constraints at the CUs are set to be identical, i.e.,  $\Gamma_{m,k}=\Gamma$, $\forall k \in {\cal K}_m, m \in {\cal L}$. The coordinates of the three BSs are set as $\left( {80~\text{m},0~\text{m}} \right)$, $\left({-40~\text{m},40\sqrt{3}~\text{m}}\right)$, and $\left( {-40~\text{m},-40\sqrt{3} ~\text{m}} \right)$, respectively. The numbers of transmit and receive antennas at each BSs are $N_t=N_r=N_a=32$. Furthermore, the path loss between each BS and CU is given by $\mu_{l,i}  = \hat \kappa  {[\frac{{{d_0}}}{{{d_{l,i}}}}]^\nu }$, where $\hat \kappa$ denotes the path loss at the reference distance of $d_0 = 1$ meter and $\nu$ denotes the path loss exponent. In addition, the targeted area is set as a square region with an area of $2 \times 2=4$ ${\text m}^2$ centering at origin $\left( {0~\text{m} ,0~\text{m}} \right)$, we take $M=9$ sample locations that are uniformly distributed in the targeted area.

\begin{figure}[htbp]
\centering
    \includegraphics[width=8cm]{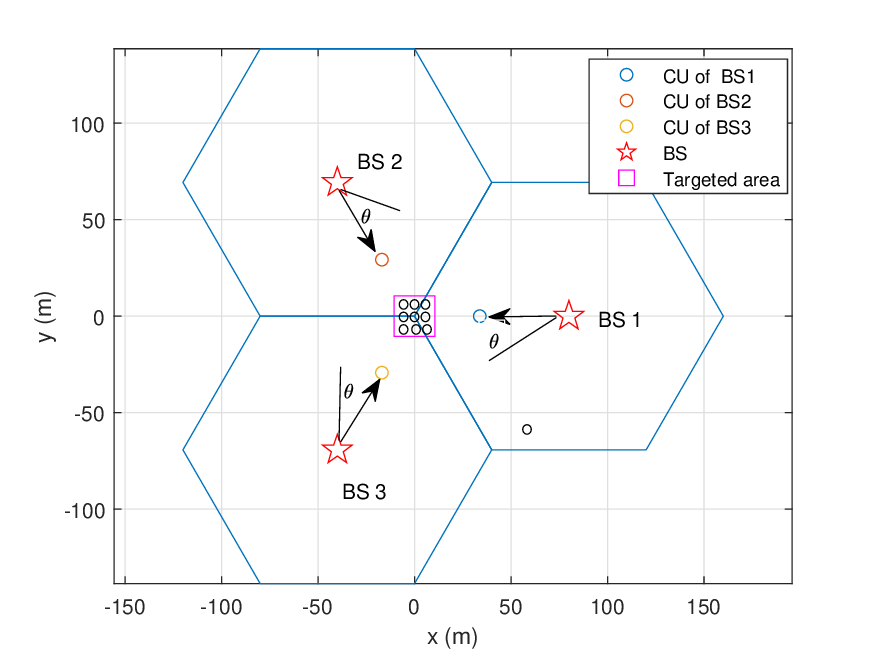}
\caption{The ISAC system has 3 BSs, each of which serves the corresponding CU for communication and detects the potential target at targeted area.}
\label{fig:2}
\end{figure}

Firstly, we consider the case when there is only $K=1$ CU served by each BS, thus the system has $LK=3$ CUs in total. In particular, we consider the Rayleigh fading channel in communication from each BS to CU, and that the three CU locations are located at  $(38.85~\text{m},-20.97~\text{m})$, $(-1.26~\text{m},44.13~\text{m})$, and $(-37.58~\text{m},-23.16~\text{m})$, respectively.

\begin{figure}[htbp]
\centering
\subfigure[Scenario I.]{\includegraphics[width=8cm]{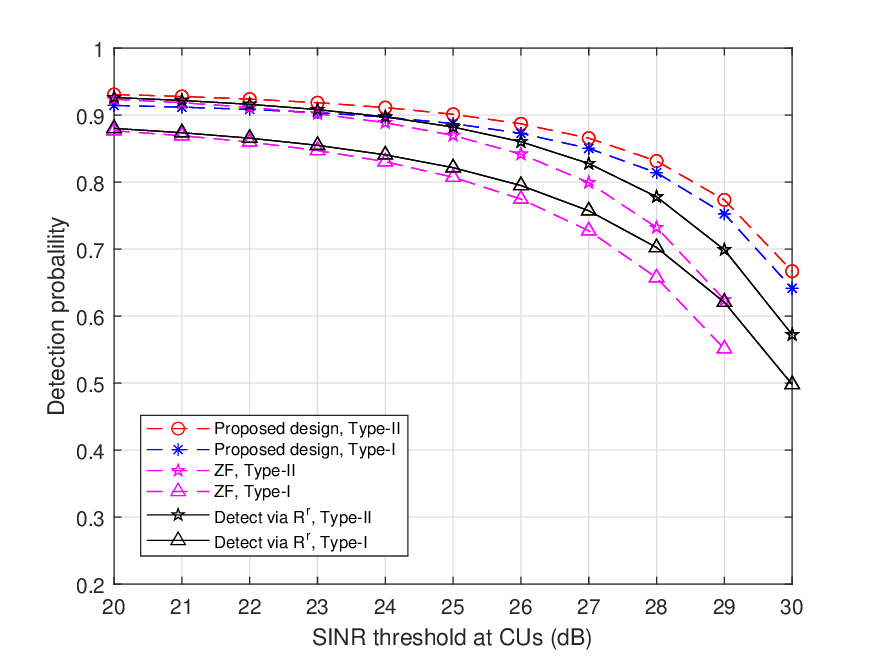}}
\subfigure[Scenario II.]{\includegraphics[width=8cm]{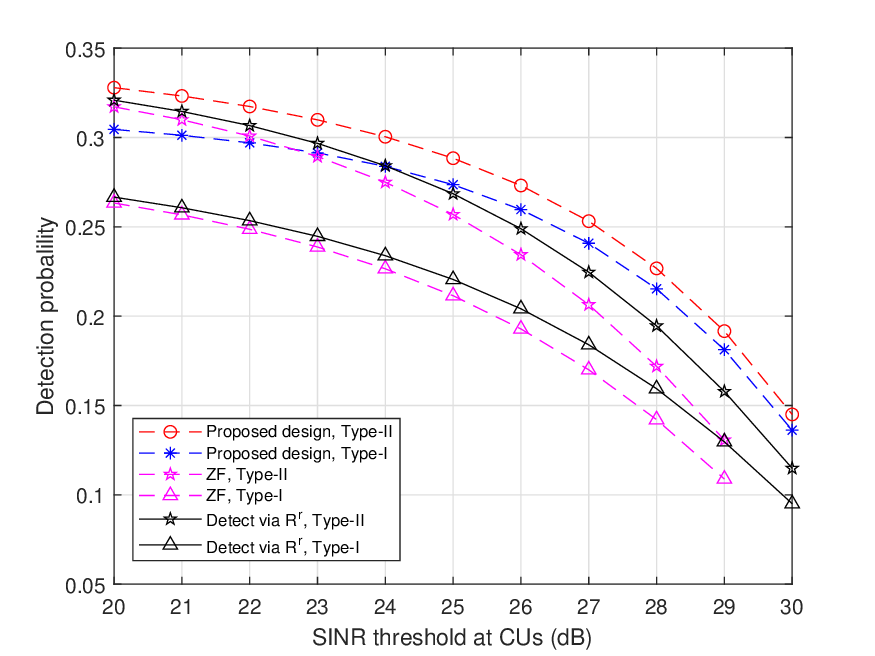}}
\caption{The detection probability versus the SINR requirement at the CUs in Scenario I and Scenario II, where $P_{\max}=15$ W, $K=1$, and $p_{\mathrm {FA}}=10^{-3}$.}
\label{fig:3}
\end{figure}

Fig. 3 shows the detection probability $p_{\mathrm D}$ versus the SINR requirement $\Gamma$ at the CUs in Scenario I and Scenario II, in which the maximum transmit power constraint $P_{\max}=15$W and the false alarm probability $p_{\mathrm {FA}}=10^{-3}$. It is observed that for all the schemes, the detection probability decreases with an increasing SINR requirement. This is due to the fact that when the communication requirement becomes stringent, the BSs need to steer the transmit beamformers towards the CUs, thus leading to less power being reflected by the target location and jeopardizing the performance of detection. It is also observed that the proposed design with Type-II receivers achieves the highest detection probability compared to other schemes under same channel conditions in both Scenario I and Scenario II. When $\Gamma$ becomes large, the performance gaps between the proposed design and the benchmark schemes are enlarged. This shows the importance of joint communication and sensing coordinated transmit beamforming. Furthermore, with a large value of $\Gamma$, the performance achieved by Type-I CU receivers approaches that of their Type-II counterparts. This is due to the fact that in this case, more power should be allocated to information signals. As a result, the power of dedicated sensing signals and the resultant interference become smaller, thus making the gain of sensing interference cancellation marginal.

\begin{figure}[htbp]
\centering
\subfigure[Scenario I.]{\includegraphics[width=8cm]{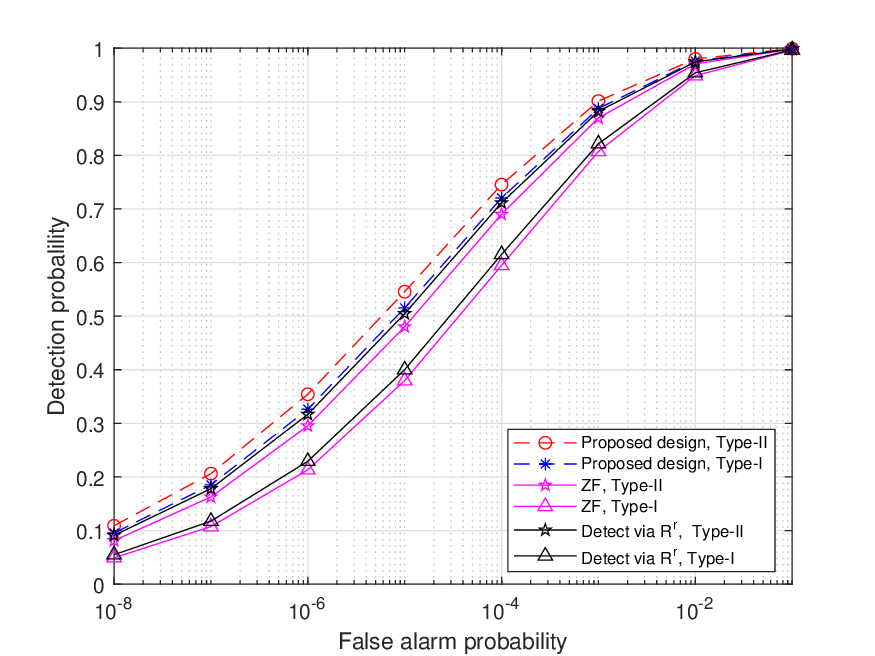}}
\subfigure[Scenario II.]{\includegraphics[width=8cm]{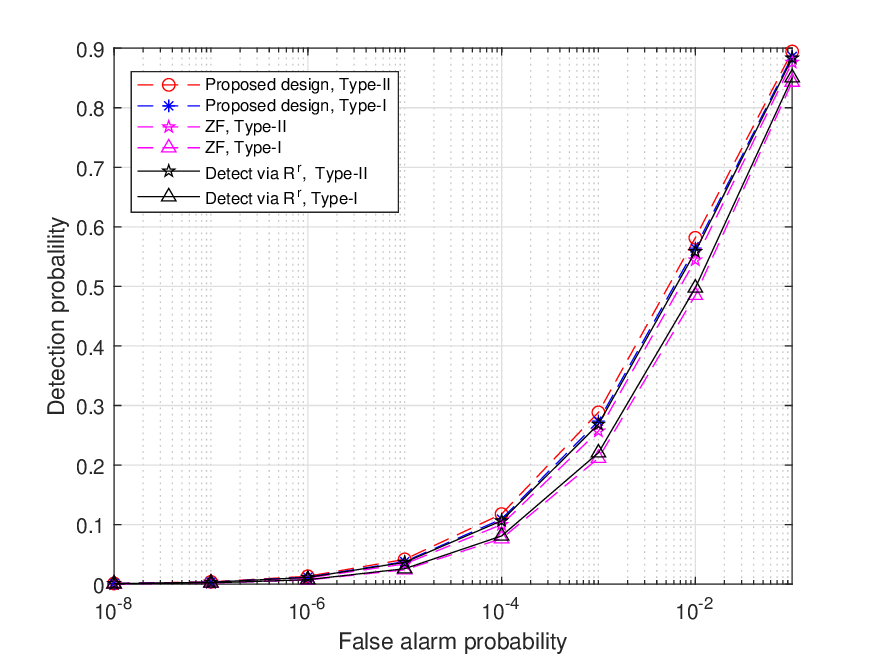}}
\caption{The detection probability versus the false alarm probability in Scenario I and Scenario II, where $P_{\max}=15$ W, $K=1$, and $\Gamma=25$ dB.}
\label{fig:4}
\end{figure}

Fig. 4 shows the detection probability $p_{\mathrm D}$ versus the false alarm probability $p_{\mathrm {FA}}$ in the two scenarios for the two types of CU receivers with $P_{\max}=15$ W and $\Gamma = 25$ dB. For all the schemes, it is observed that the detection probability increases towards one as the false alarm probability becomes large, as correctly predicted by \eqref{pdI} and \eqref{pdII}. It is also observed that the detection probability achieved in Scenario II is much smaller than that in Scenario I under the same setup. This gain is attributed to the joint exploitation of both direct and cross echo links in Scenario I, thanks to the time synchronization among different BSs.  Furthermore, the schemes with Type-II receivers are observed to achieve higher detection probability than the counterparts with Type-I receivers, thus validating again the benefit of dedicated sensing signals together with interference cancellation in enhancing the networked ISAC performance.

\begin{figure}[htbp]
\centering
\subfigure[Scenario I.]{\includegraphics[width=8cm]{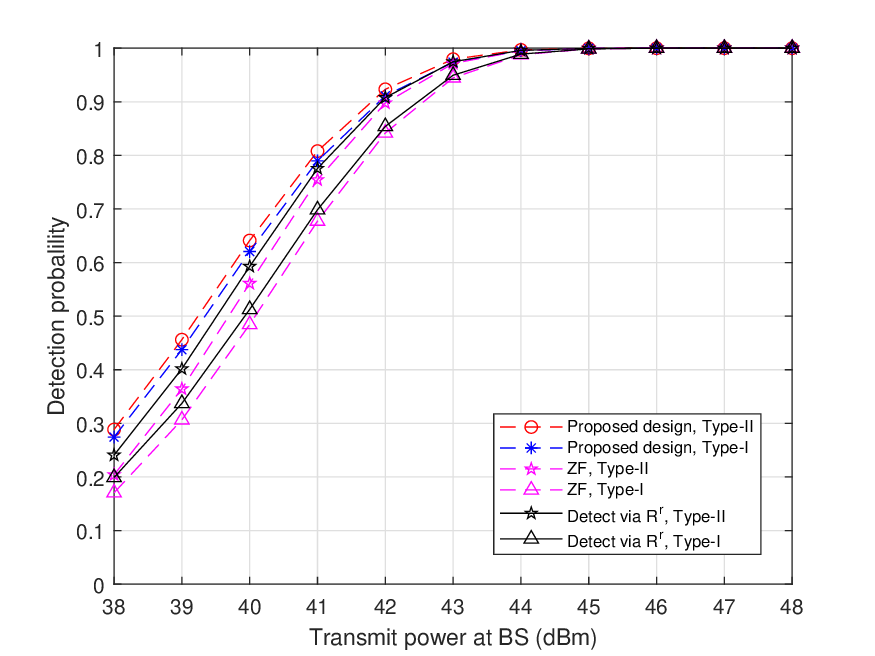}}
\subfigure[Scenario II.]{\includegraphics[width=8cm]{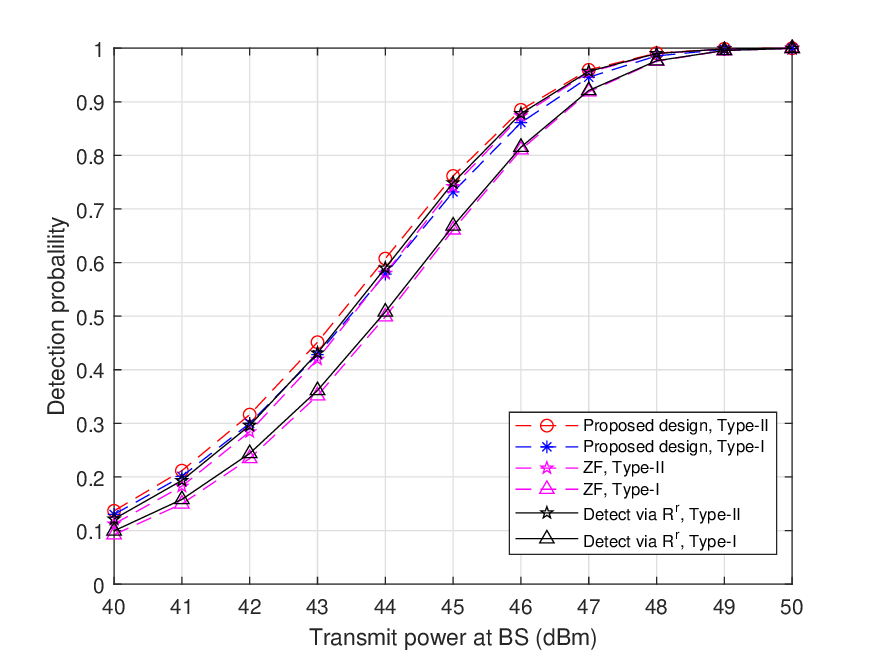}}
\caption{The detection probability versus the transmit power at each BSs in Scenario I and Scenario II, where $\Gamma=25$ dB, $K=1$, and $p_{\mathrm {FA}}=10^{-3}$.}
\label{fig:5}
\end{figure}

Fig. 5 shows the detection probability $p_{\mathrm D}$ versus the maximum transmit power budget $P_{\max}$ at each BS with $\Gamma = 25$ dB and $p_{\mathrm {FA}}=10^{-3}$. Similar observations can be made as in Figs. 3 and 4, demonstrating the performance gains achieved by the proposed joint optimization framework.

%It is observed that the proposed scheme with Type-II receivers has higher detection performance than other schemes in both Scenario I and Scenario II. The schemes of joint detection via dedicated sensing signals achieve a lower detection probability than the proposed scheme. This is due to the fact that it only applies dedicated sensing signals for detection and can not utilize the communication signals, which is also beneficial to detection because it enhances the echo power. The zero-forcing design with Type-I receivers achieves the lowest detection probability. This illustrates again the importance of having a proper coordinated transmit beamforming design in a networked ISAC system.

Next, we consider the case when each CU is located at a different angle with respect to its home BS, in which the LoS channel is considered in communication from each BS to its respective CU. In particular, as shown in Fig. 2, the CUs in different cells are located at the same angle $\theta$ with respect to the corresponding BS, and the distance between each BS and the correspondingly associated CU is $45~\text{m}$.
\begin{figure}[htbp]
\centering
\subfigure[Scenario I.]{\includegraphics[width=8cm]{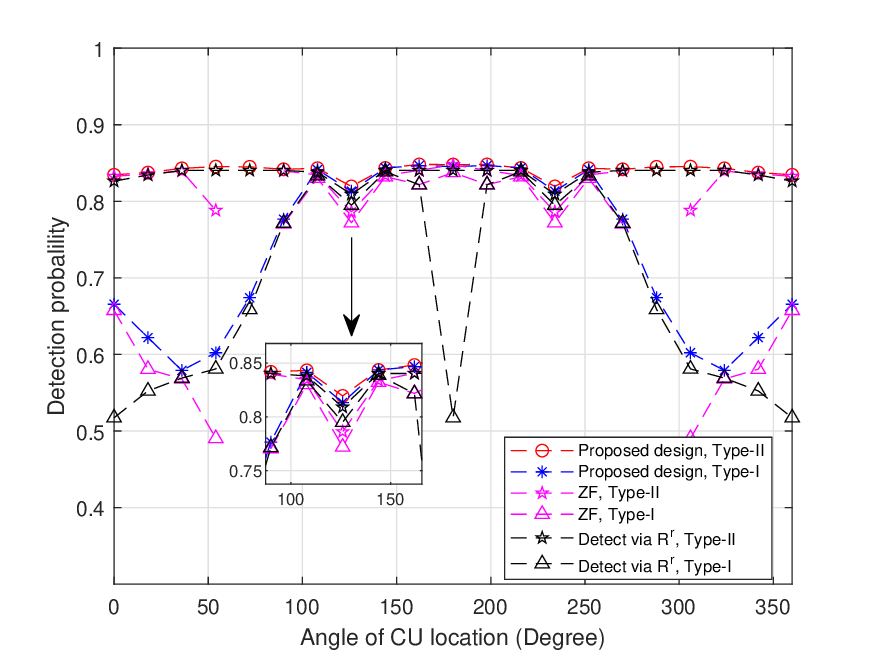}}
\subfigure[Scenario II.]{\includegraphics[width=8cm]{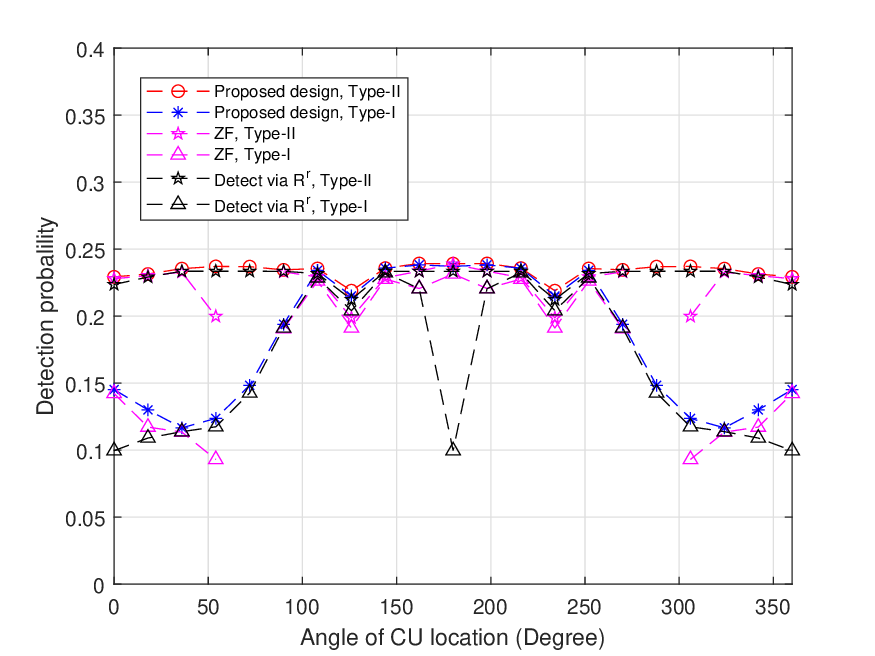}}
\caption{The detection probability versus the angle $\theta$ at which the CU position has been rotated, where $P_{\max}=12$ W, $\Gamma=30$ dB, $K=1$, and $p_{\mathrm {FA}}=10^{-3}$.}
\label{fig:7}
\end{figure}

Fig. 6 shows the detection probability $p_{\mathrm D}$ versus the angle $\theta$ with $P_{\max}=12$ W, $\Gamma=30$ dB, and $p_{\mathrm {FA}}=10^{-3}$. It is observed that when $\theta$ is close to  ${0^ \circ }$ and ${360^ \circ }$, the scheme with Type-II receivers significantly outperforms that with Type-I receivers. This is because the CUs are located at similar angles as the targeted area in this case, and accordingly, the interference caused by sensing signals becomes severe. By contrast, when $\theta$ is close to $110^\circ$, such performance gap is observed to become less. This is due to the fact that the CUs are located at different angles from the targeted area. As a result, the information and sensing beamformers can be steered toward different directions for communication and sensing, respectively, with minimized interference.

 %in which the CUs are close to the target (or even at same direction to target), the interference caused by the dedicated sensing signals becomes the main reason of performance loss of the schemes with Type-I receivers. It is also observed that when $\theta$ approaches $180^{\circ}$, in which the CUs and target are at the opposite direction, the scheme of joint detection via $R^{r}$ with Type-I receivers also suffers a certain performance loss. One reason is that the BSs have to steer the beam to the opposite direction and only the dedicated sensing signals can be efficiently exploited in this scheme. Also, the other reason is that the BSs have to reduce the power allocated to the dedicated sensing signals in order to suppress the interference for communication, thus decrease the detection probability. The above conclusions show that our proposed scheme can achieve better communication and sensing trade-offs.

\begin{figure}[!t]
\centering
    \includegraphics[width=8cm]{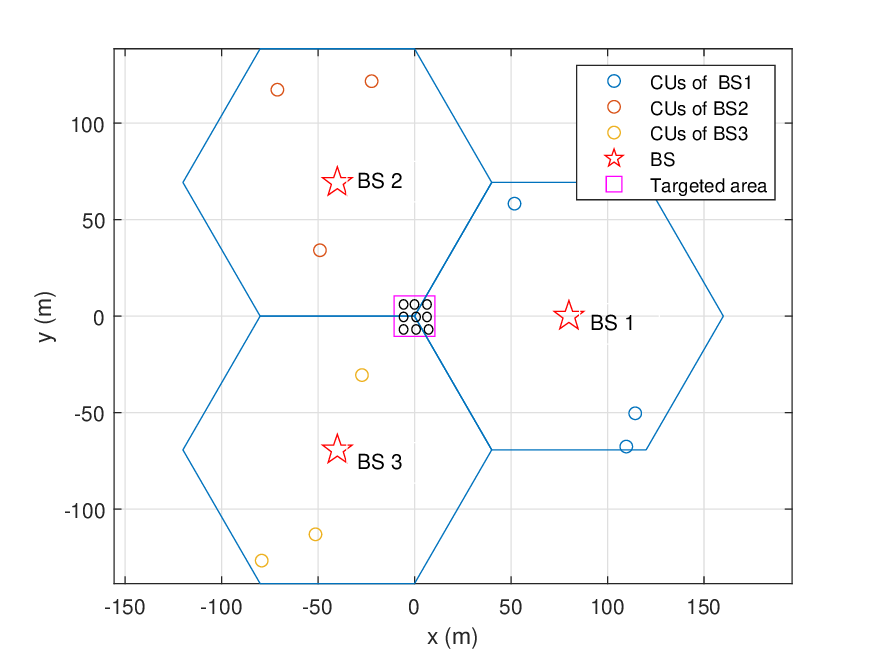}
\caption{Networked ISAC topology with 3 BSs each serves 3 CUs (CU locations are randomly generated).}
\label{fig:7}
\end{figure}

\begin{figure}[htbp]
\centering
\subfigure[Scenario I.]{\includegraphics[width=8cm]{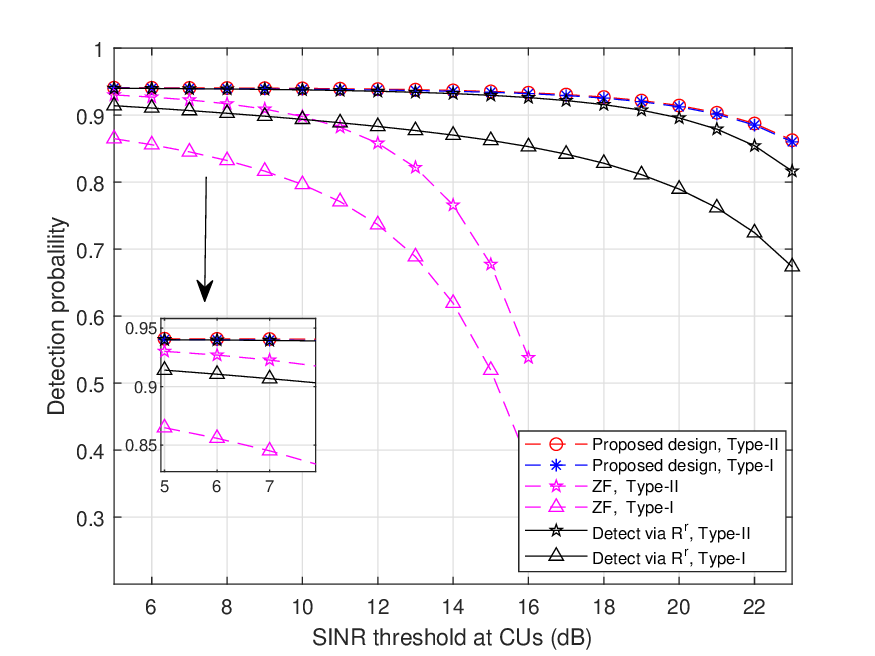}}
\subfigure[Scenario II.]{\includegraphics[width=8cm]{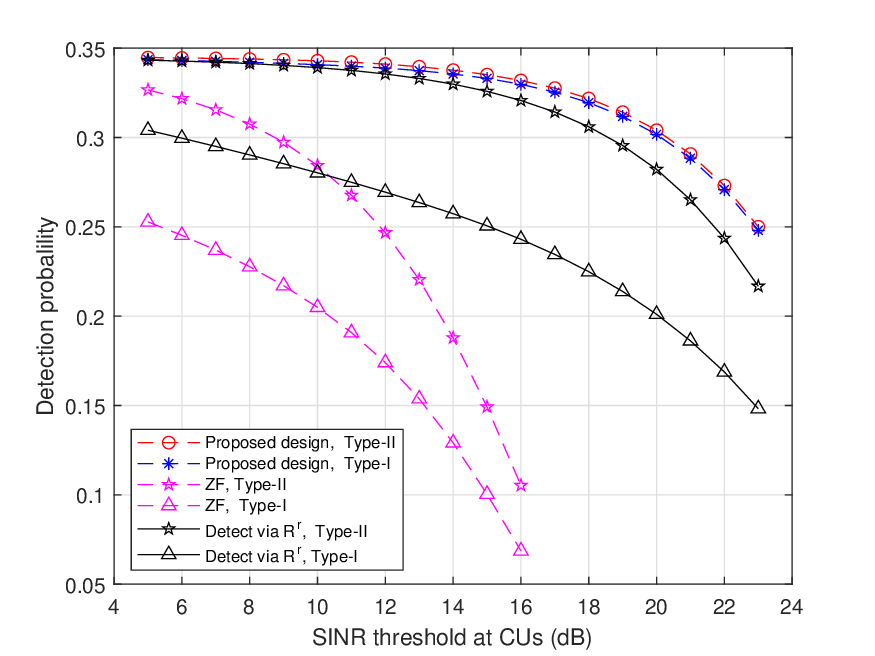}}
\caption{The detection probability versus the SINR requirement at CUs in Scenario I and Scenario II, where $P_{\max}=15$W, $\Gamma=15$ dB, $K=3$, and $p_{\mathrm {FA}}=10^{-3}$.} \label{fig:8}
\end{figure}

\begin{figure}[htbp]
\centering
\subfigure[Scenario I.]{\includegraphics[width=8cm]{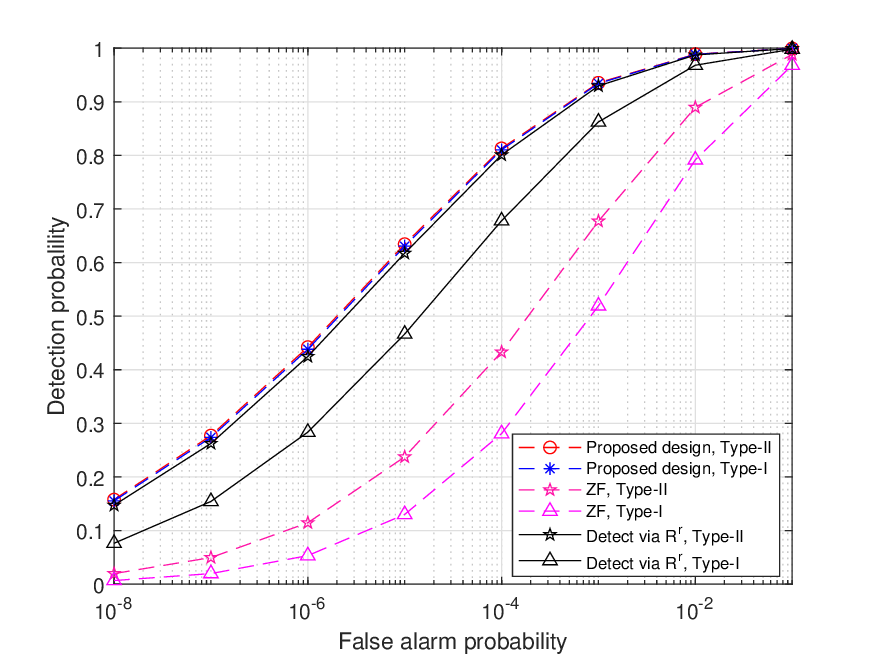}}
\subfigure[Scenario II.]{\includegraphics[width=8cm]{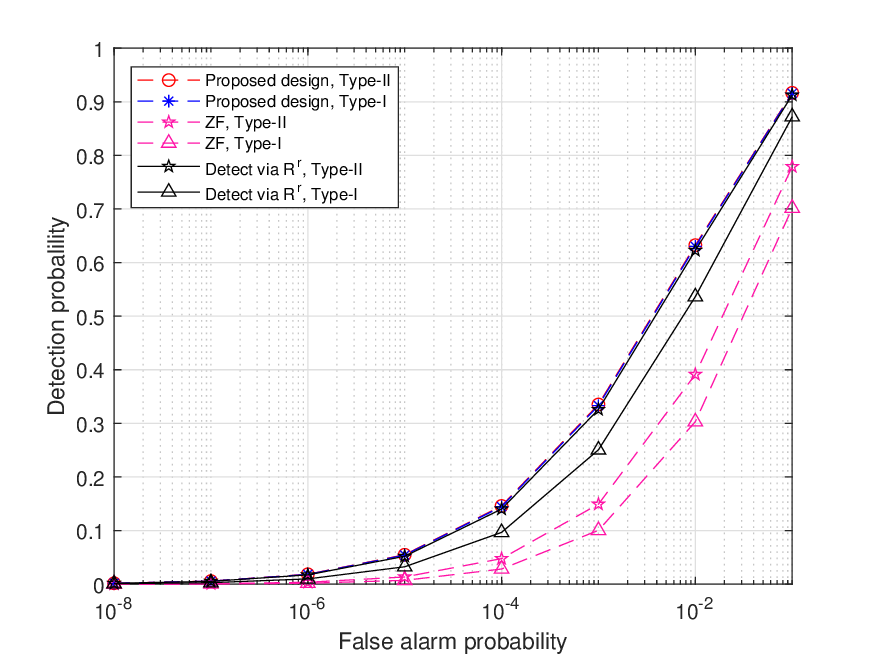}}
\caption{The detection probability versus the false alarm probability in Scenario I and Scenario II, where $P_{\max}=15$ W, $\Gamma=15$ dB, $K=3$, and $p_{\mathrm {FA}}=10^{-3}$.} \label{fig:9}
\end{figure}

\begin{figure}[htbp]
\centering
\subfigure[Scenario I.]{\includegraphics[width=8cm]{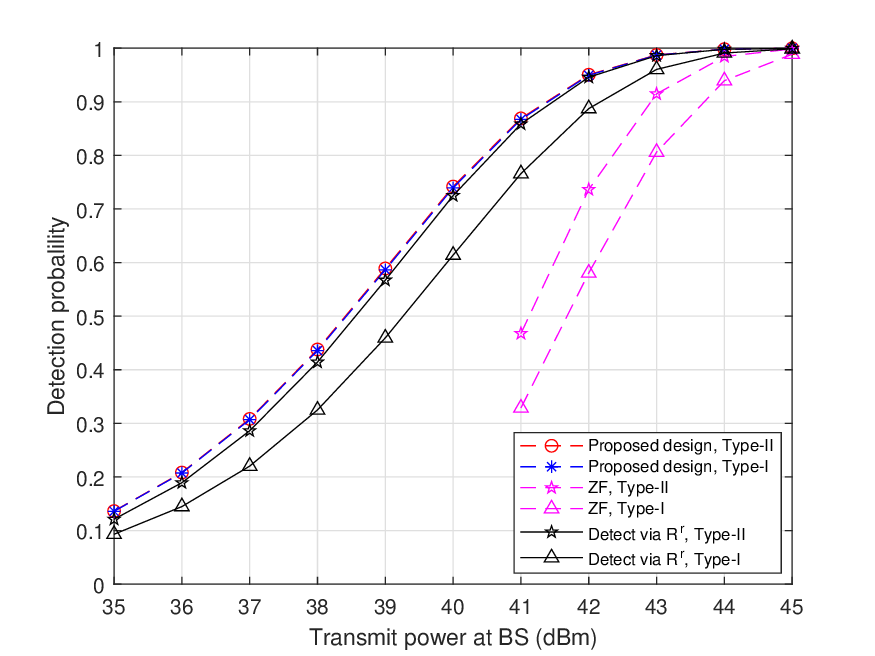}}
\subfigure[Scenario II.]{\includegraphics[width=8cm]{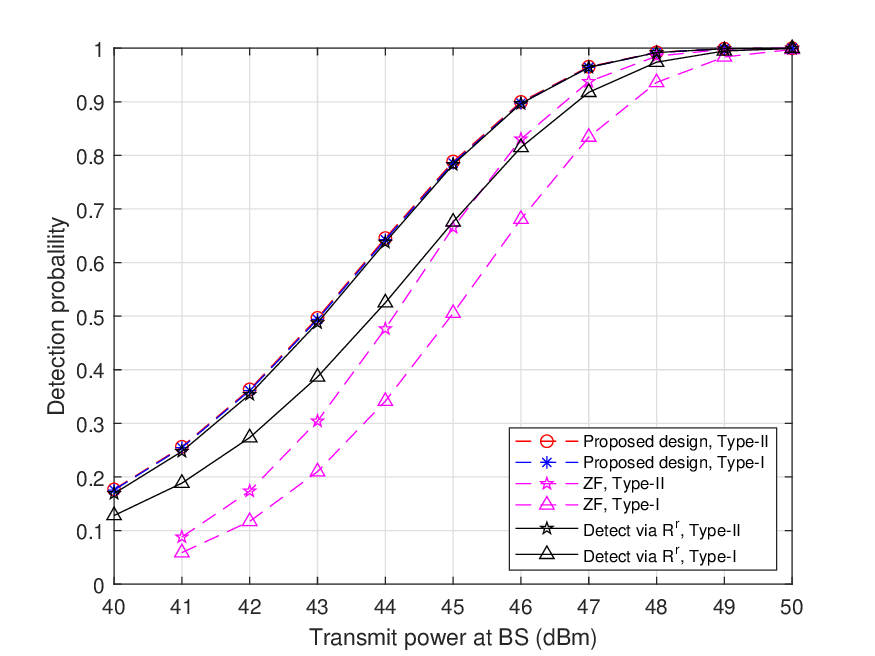}}
\caption{The detection probability versus the transmit power in Scenario I and Scenario II, where $\Gamma=15$ dB, $K=3$, and $p_{\mathrm {FA}}=10^{-3}$.} \label{fig:10}
\end{figure}

Furthermore, we consider that each BS serves $K=3$ CUs, in which the Rayleigh fading channel model is considered from each BS to its respective CU. The locations of the CUs are generated as shown in Fig. 7. Fig. 8 shows the detection probability $p_{\mathrm D}$ versus the SINR requirement $\Gamma$ at each CU, where $P_{\max}=15$ W, $\Gamma=15$ dB, $p_{\mathrm {FA}}=10^{-3}$. Fig. 9 shows the detection probability $p_{\rm D}$ versus the false alarm probability $p_{\rm{FA}}$ with $P_{\max} = 15$ W,  $\Gamma = 15$ dB, and $p_{\rm{FA}} = 10^{-3}$. Fig. 10 shows the detection probability $p_{\mathrm D}$ versus the power budget $P_{\max}=15$ at each BS, with $\Gamma = 15$ dB, and $p_{\rm{FA}} = 10^{-3}$. By comparing these figures to Figs. 3, 4, and 5 for the case with $K=1$, it is observed that the proposed design with Type-I receivers achieves similar performance as that with Type-II receivers, which means that the gain brought by the dedicated signals becomes marginal in this case. This is due to the fact that when there are more CUs in each cell, we have a larger number of information beams that can provide sufficient degrees of freedom for target sensing, thus making the benefit of dedicated sensing signals limited or even not necessary. It is also observed that the benchmark scheme with ZF-based information beamforming performs significantly worse than the proposed designs. This is because with more CUs, the inter-user interference becomes more severe, and thus the ZF-based design leads to degraded performance due to the limited available degrees of freedom for effective interference suppression.

\section{Conclusion}
This paper studied the joint multi-cell communication and distributed MIMO radar detection in a networked ISAC system, in which a set of multi-antenna BSs employed the coordinated transmit beamforming to serve their associated single-antenna CUs, and at the same time utilized the dedicated sensing signals together with their communication signals for target detection. Two joint detection scenarios with and without time synchronization among the BSs were considered, for which the detection probability and the false alarm probability were derived in closed forms. Accordingly, we developed the coordinated transmit beamforming ISAC design to maximize the minimum detection probability (or equivalently the total received reflection-signal power) over a particular targeted area, while ensuring the SINR constraints at the CUs for communication. By considering the transmission of dedicated sensing signals, we introduced two types of CU receivers, Type-I and Type-II, without and with the capability of dedicated sensing interference cancellation, respectively. For the proposed non-convex optimization problems, we adopted the SDR technique to obtain the optimal joint beamforming solutions. Finally, numerical results showed that the proposed ISAC design achieved higher detection probability than other benchmark schemes. It was also shown that the presence of time synchronization among the BSs and dedicated sensing interference cancellation can further enhance the sensing and communication performances for networked ISAC systems. %As the ISAC systems become more complex and the task facing becomes more stringent, the investigation of  networked ISAC is a important and urgent topic for future research.

\appendix
\subsection{Proof of Proposition 1} \label{appendixA}
It can be verified based on \eqref{app:1:1} and \eqref{app:1:3} that ${{\bf{\tilde W}}}_{l,i}$ achieves the same objective values in (P1.1) as ${\bf{W}}_{l,i}^ *$, and satisfy the power constraints in \eqref{prob1.1:con:2}.

Next, we verify that ${{\bf{\tilde W}}}_{l,i}$ can satisfy the SINR constraints in \eqref{prob1.1:con:2} for communications. From (\ref{app:1:1}) and \eqref{app:1:2}, we obtain that
\begin{align}
{\bf{h}}_{l,m,k}^H{{{\bf{\tilde W}}}_{l,i}}{{\bf{h}}_{l,m,k}} &= {\bf{h}}_{l,m,k}^H{{{\bf{\tilde w}}}_{l,i}}{\bf{\tilde w}}_{l,i}^H{{\bf{h}}_{l,m,k}} \nonumber \\ &= {\bf{h}}_{l,m,k}^H{\bf{W}}_{l,i}^ * {{\bf{h}}_{l,m,k}}. \label{app:1:5}
\end{align}
Thus, it follows that
\begin{align}
&( {1 + \frac{1}{{\Gamma _{m,k}}}} ){\bf{h}}_{m,m,k}^H{\bf{\tilde W}}_{m,k}{{\bf{h}}_{m,m,k}}\nonumber \\
&= ( {1 + \frac{1}{{\Gamma _{m,k}}}} ){\bf{h}}_{m,m,k}^H{\bf{W}}_{m,k}^ * {{\bf{h}}_{m,m,k}} \nonumber \\
&\ge \sum\limits_{l \in {\cal L}} {{\bf{h}}_{l,m,k}^H( {\sum\limits_{i \in {\cal K}_l} {{\bf{W}}_{l,i}^*}  + {\bf{R}}_l^{r*}} ){{\bf{h}}_{l,m,k}}}  + {\sigma_c ^2} \nonumber \\
&= \sum\limits_{l \in {\cal L}} {{\bf{h}}_{l,m,k}^H( {\sum\limits_{i \in {\cal K}_l} {{\bf{\tilde W}}_{l,i}}  + {\bf{\tilde R}}_l^{r}} ){{\bf{h}}_{l,m,k}}}  + {\sigma_c ^2}. \label{qos:1.1}
\end{align}
The first equality follows from \eqref{app:1:5}, the inequality follows from  \eqref{prob1.1:con:1}, and the last equality follows from \eqref{app:1:3}. Therefore, we show that the constructed solution $\{{\bf{\tilde W}}_{l,i}\}$ and $\{{\bf{\tilde R}}_l^r\}$ also satisfy the SINR constraints in \eqref{prob1.1:con:2} for  problem (P1.1). Thus, this completes the proof of this proposition.

\bibliographystyle{IEEEtran}
\bibliography{IEEEabrv,myref}

\end{document}